\documentclass[10pt]{article} %
\usepackage[preprint]{tmlr}
\usepackage[figuresleft]{rotating}
\usepackage{hyperref}
\usepackage{url}
\usepackage{color,soul}
\usepackage{epsf}
\usepackage{graphicx}
\usepackage{palatino}
\usepackage{amsmath,amssymb,amsthm}
\usepackage{nicefrac}
\newcommand{\newc}{\newcommand}
\newc{\E}{\mathbb{E}}
\newc{\V}{\mbox{V}}
\newc{\N}{\mbox{N}}
\newc{\Bern}{\mbox{Bern}}
\newc{\Po}{\mbox{Po}}
\newc{\IG}{\mbox{IG}}
\newc{\Gam}{\mbox{Gam}}
\newc{\bdp}{\mathbf{p}}
\newc{\bdt}{\mathbf{t}}
\newc{\Normal}{\mathcal{N}}
\newc{\Expectation}{\mathbb{E}}
\newc{\odds}{\mbox{odds}}
\newc{\stderr}{\mbox{s.e.}}
\newc{\logit}{\mbox{logit}}
\newc{\sign}{\mbox{sign}}
\newc{\SD}{\mbox{SD}}
\newc{\qtilde}{\tilde{q}}
\newc{\deltahat}{\hat{\delta}}
\newc{\shat}{\hat{s}}
\newc{\bbar}{\bar{b}}
\newc{\tauhat}{\hat{\tau}}
\newc{\bdq}{\mathbf{q}}
\newc{\bdmu}{\mbox{\boldmath $\mu$}}
\newc{\bdSigma}{\mbox{\boldmath $\Sigma$}}
\newc{\bdLambda}{\mbox{\boldmath $\Lambda$}}
\newc{\bdmuhat}{\mbox{\boldmath $\hat{\mu}$}}
\newc{\bdeta}{\mbox{\boldmath $\eta$}}
\newc{\bdtheta}{\mbox{\boldmath $\theta$}}
\newc{\bdbeta}{\mbox{\boldmath $\beta$}}
\newc{\bdgamma}{\mbox{\boldmath $\gamma$}}
\newc{\bdbetahat}{\mbox{\boldmath $\hat{\beta}$}}
\newc{\bdgammahat}{\mbox{\boldmath $\hat{\gamma}$}}
\newc{\bdthetahat}{\mbox{\boldmath $\hat{\theta}$}}
\newc{\bdvareps}{\mbox{\boldmath $\varepsilon$}}
\newc{\bdzero}{\mbox{\boldmath $0$}}
\newc{\bdone}{\mbox{\boldmath $1$}}
\newc{\bdnu}{\mbox{\boldmath $\nu$}}
\newc{\bdell}{\mbox{\boldmath $\ell$}}
\newc{\bdxi}{\mbox{\boldmath $\xi$}}
\newc{\bdomega}{\mbox{\boldmath $\omega$}}
\newc{\bdepsilon}{\mbox{\boldmath $\varepsilon$}}
\newc{\bdI}{\mathbf{I}}
\newc{\bdP}{\mbox{\boldmath $P$}}
\newc{\bdX}{\mbox{\boldmath $X$}}
\newc{\bdA}{\mbox{\boldmath $A$}}
\newc{\bdB}{\mbox{\boldmath $B$}}
\newc{\bdC}{\mbox{\boldmath $C$}}
\newc{\bdD}{\mbox{\boldmath $D$}}
\newc{\bdG}{\mbox{\boldmath $G$}}
\newc{\bdJ}{\mbox{\boldmath $J$}}
\newc{\bdK}{\mbox{\boldmath $K$}}
\newc{\bda}{\mbox{\boldmath $a$}}
\newc{\bdb}{\mbox{\boldmath $b$}}
\newc{\bdc}{\mbox{\boldmath $c$}}
\newc{\bde}{\mbox{\boldmath $e$}}
\newc{\bdu}{\mbox{\boldmath $u$}}
\newc{\bdv}{\mbox{\boldmath $v$}}
\newc{\bdx}{\mbox{\boldmath $x$}}
\newc{\bdy}{\mbox{\boldmath $y$}}
\newc{\bdz}{\mbox{\boldmath $z$}}
\newc{\bdr}{\mbox{\boldmath $r$}}
\newc{\bdQ}{\mbox{\boldmath $Q$}}
\newc{\bdR}{\mbox{\boldmath $R$}}
\newc{\bdY}{\mbox{\boldmath $Y$}}
\newc{\bdT}{\mbox{\boldmath $T$}}
\newc{\bdW}{\mbox{\boldmath $W$}}
\newc{\bdH}{\mbox{\boldmath $H$}}
\newc{\bdL}{\mbox{\boldmath $L$}}
\newc{\bdU}{\mbox{\boldmath $U$}}
\newc{\bdV}{\mbox{\boldmath $V$}}
\newc{\Multinom}{\mbox{Multinom}}
\newc{\tbar}{\bar{t}}
\newc{\Var}{\mbox{Var}}
\newc{\var}{\mbox{var}}
\newc{\diag}{\mbox{diag}}
\newc{\tr}{\mbox{tr}}
\newc{\phat}{\hat{p}}
\newc{\Xbar}{\bar{X}}
\newc{\xbar}{\bar{x}}
\newc{\Ybar}{\bar{Y}}
\newc{\ybar}{\bar{y}}
\newc{\dbar}{\bar{d}}
\newc{\yhat}{\hat{y}}
\newc{\bdyhat}{\mbox{\boldmath $\hat{y}$}}
\newc{\ytil}{\tilde{y}}
\newc{\ytilde}{\tilde{y}}
\newc{\ftil}{\tilde{f}}
\newc{\Ho}{\mbox{\bf H}_o}
\newc{\Ha}{\mbox{\bf H}_a}
\newc{\phatYX}{\phat_Y - \phat_X}
\newc{\SSG}{\mbox{SSG}}
\newc{\SSB}{\mbox{SSB}}
\newc{\SSE}{\mbox{SSE}}
\newc{\SST}{\mbox{SST}}
\newc{\SSR}{\mbox{SSR}}
\newc{\SSAB}{\mbox{SSAB}}
\newc{\MSG}{\mbox{MSG}}
\newc{\MSB}{\mbox{MSB}}
\newc{\MSE}{\mbox{MSE}}
\newc{\MST}{\mbox{MST}}
\newc{\MSAB}{\mbox{MSAB}}
\newc{\dfE}{\mbox{dfE}}
\newc{\dfG}{\mbox{dfG}}
\newc{\dfB}{\mbox{dfB}}
\newc{\dfT}{\mbox{dfT}}
\newc{\dfAB}{\mbox{dfAB}}
\newc{\muhat}{\hat{\mu}}
\newc{\betahat}{\hat{\beta}}
\newc{\alphahat}{\hat{\alpha}}
\newc{\etahat}{\hat{\eta}}
\newc{\phihat}{\hat{\phi}}
\newc{\sigmahat}{\hat{\sigma}}
\newc{\cl}{\centerline}
\newc{\redtitle}[1]{ {\color{red}\und{#1}:} }
\newc{\bluetitle}[1]{ {\color{blue}\und{#1}:} }
\newc{\magentatitle}[1]{ {\color{magenta}\und{#1}:} }
\newc{\R}{\mathbb{R}}
\newc{\trans}{^\mathsf{T}}
\newc{\xtx}{\bdX\trans\bdX}
\newc{\xxtxx}{\bdX(\xtx)^{-1}\bdX\trans}
\newc{\argmin}{\operatornamewithlimits{argmin}}
\newc{\argmax}{\operatornamewithlimits{argmax}}
\newc{\setS}{\mathcal{S}}
\newc{\toP}{\overset{p}{\to}}
\newc{\toD}{\overset{d}{\to}}
\newc{\eqD}{\overset{d}{=}}
\newc{\toAS}{\overset{a.s.}{\to}}
\newc{\simIID}{\overset{\text{i.i.d.}}{\sim}}
\newc{\simIND}{\overset{\text{ind.}}{\sim}}
\newc{\thetahat}{\hat{\theta}}
\newc{\bds}{\mathbf{s}}
\newc{\thetatilde}{\tilde{\theta}}
\newc{\Phat}{\hat{P}}
\newc{\Qhat}{\hat{Q}}
\newc{\tny}{\small}
\newc{\Col}{\textnormal{Col}}
\newc{\mutilde}{\tilde{\mu}}
\newc{\vecst}{\textnormal{vec}}
\newc{\indep}{\perp \!\!\! \perp}
\newc{\Cov}{\textnormal{Cov}}
\newc{\gtilde}{\tilde{g}}
\newc{\rank}{\textnormal{rank}}
\newc{\Cor}{\textnormal{Cor}}
\newc{\Lcol}{\mathcal{L}_{col}}
\newc{\Lrow}{\mathcal{L}_{row}}
\newc{\Ell}{\mathcal{L}}
\newc{\Norm}{\mathcal{N}}
\newc{\xtilde}{\tilde{x}}
\newc{\Xtilde}{\tilde{X}}
\newc{\Atilde}{\tilde{A}}
\newc{\Xhat}{\hat{X}}
\newc{\Vhat}{\hat{V}}
\newc{\Vperp}{V^\perp}
\newc{\Rhat}{\hat{R}}
\newc{\Shat}{\hat{S}}
\newc{\uhat}{\hat{u}}
\newc{\hhat}{\hat{h}}
\newc{\Yhat}{\hat{Y}}
\newc{\bpage}{\eject}
\newc{\bdg}{\mathbf{g}}
\newc{\bdZ}{\mathbf{Z}}
\newc{\Rn}{\mathbb{R}^n}
\newc{\Sigmahat}{\hat{\Sigma}}
\newc{\Ysubj}{Y_{(j)}}
\newc{\Rsubj}{R_{(j)}}
\newc{\A}{\mathcal{A}}
\newc{\B}{\mathcal{B}}
\newc{\Gammasubj}{\Gamma_{(j)}}
\newc{\Btilde}{\tilde{\mathcal{B}}}
\newc{\Xitilde}{\tilde{\Xi}}
\newc{\xitilde}{\tilde{\xi}}
\newc{\Gammatilde}{\tilde{\Gamma}}
\newc{\gammatilde}{\tilde{\gamma}}
\newc{\Lambdatilde}{\tilde{\Lambda}}
\newc{\Ls}{\mathcal{L}}
\newc{\Lambdasubj}{\Lambda_{(j)}}
\newc{\Msubj}{M_{(j)}}
\newc{\Lsubj}{L_{(j)}}
\newc{\Ytilde}{\tilde{Y}}
\newc{\Ytildehat}{\hat{\tilde{Y}}}
\newc{\xtildebar}{\tilde{\xbar}}
\newc{\F}{\mathcal{F}}
\newc{\M}{\mathcal{M}}
\newc{\G}{\mathcal{G}}
\newc{\Dtilde}{\tilde{D}}
\newc{\Utilde}{\tilde{U}}
\newc{\Vtilde}{\tilde{V}}
\newc{\betatilde}{\tilde{\beta}}
\newc{\fhat}{\hat{f}}
\newc{\Rtilde}{\tilde{R}}
\newc{\Qtilde}{\tilde{Q}}
\newc{\gammahat}{\hat{\gamma}}
\newc{\X}{\mathcal{X}}
\newc{\Y}{\mathcal{Y}}
\newc{\abar}{\bar{a}}
\newc{\phitilde}{\tilde{\phi}}
\newc{\ahat}{\hat{a}}
\newc{\bhat}{\hat{b}}
\newc{\bdf}{\mathbf{f}}
\newc{\bdh}{\mathbf{h}}
\newc{\Pois}{\text{Pois}}
\newc{\bdlambda}{\boldsymbol{\lambda}}
\newc{\Mult}{\text{Mult}}
\newc{\liminfc}{\underline{\lim}}
\newc{\limsupc}{\overline{\lim}}
\newc{\pbar}{\overline{p}}
\newc{\bdalpha}{\boldsymbol{\alpha}}
\newc{\psibar}{\overline{\psi}}
\newc{\bdxbar}{\overline{\mathbf{x}}}
\newc{\bdd}{\mathbf{d}}
\newc{\Wtilde}{\tilde{W}}
\newc{\bdZtilde}{\tilde{\mathbf{Z}}}
\newc{\ttilde}{\tilde{t}}
\newc{\Ztilde}{\tilde{Z}}
\newc{\pihat}{\hat{\pi}}
\newc{\betahatridge}{\betahat^{\lambda}_{\text{RR}}}
\newc{\epsilonhat}{\hat{\epsilon}}
\newc{\zhat}{\hat{z}}
\newc{\qhat}{\hat{q}}
\newc{\Bhat}{\hat{B}}

\newc{\hh}{\mathcal{H}}
\newc{\Unif}{\text{Unif}}
\newc{\psihat}{\hat{\psi}}
\newc{\Expo}{\text{Expo}}
\newc{\bdxtilde}{\tilde{\mathbf{x}}}
\newc{\toDPn}{\overset{P_n}{\rightsquigarrow}}
\newc{\toDQn}{\overset{Q_n}{\rightsquigarrow}}
\newc{\Fhat}{\hat{F}}
\newc{\bdO}{\mathbf{O}}
\newc{\bdn}{\mathbf{n}}
\newc{\Uhat}{\hat{U}}
\usepackage{float}
\usepackage{booktabs}
\newtheorem{theorem}{Theorem}
\newtheorem{assumption}[theorem]{Assumption}
\usepackage[section]{placeins}

\newcommand{\Ahat}{\hat{A}}
\newcommand{\Chat}{\hat{C}}
\newtheorem{proposition}[theorem]{Proposition}

\title{Pro-Active Inference}

\author{\name Skyler Wu \email skylerw@stanford.edu \\
      \addr Department of Statistics\\Stanford University\\
      \AND
      \name Yash Nair \email yashnair@stanford.edu \\
      \addr Department of Statistics\\Stanford University\\
      \AND
      \name Antoine Maechler \email maechler@stanford.edu\\
      \addr Department of Management Science and Engineering\\Stanford University
      \AND
      \name Emmanuel J. Candès \email candes@stanford.edu\\
      \addr Department of Statistics\\Stanford University
}

\def\month{10}  %
\def\year{2026} %
\def\openreview{\url{https://openreview.net/forum?id=XXXX}} %

\begin{document}

\maketitle

\begin{abstract}
Estimating a population mean is fundamental across scientific applications, but observing every outcome can be prohibitively expensive. Often, covariates and machine learning (ML) predictions are available for the full population, even when only a small subset of outcomes can be observed. How can we leverage this information to proactively select observations to label, update ML predictions as labels arrive, and construct statistically valid confidence intervals? We develop a general, without-replacement formulation of Pro-Active Inference (PAI), extending our earlier with-replacement approach and building on Active Statistical Inference. Rather than making fixed-order label-skip decisions, PAI proactively selects which observation to label next from the remaining covariate pool while accommodating online ML updates. We establish PAI's finite-sample unbiasedness and asymptotic validity, develop an improved finite-sample variance estimator, and extend PAI to superpopulation inference without Bonferroni corrections or $\alpha$-splitting. We characterize PAI's potential efficiency gains over Batch Active Inference in a quasi-oracle setting, highlighting the role of prediction-residual imbalance. Across eight diverse datasets, PAI achieves greater sampling efficiency than Sequential Active Inference in the vast majority of tested settings with near-nominal statistical coverage, using the same ML predictor and uncertainty score. However, comparisons with simpler baselines such as Prediction-Powered Inference and uniform sampling, as well as the closely related Prediction-Powered Active Testing, reveal when additional adaptivity helps or hurts statistical efficiency. Through interpretable, theory-motivated diagnostics and frozen-ML ablations, we investigate when and why PAI gains or loses efficiency, highlighting the initial residual-estimation problem and subsequent online learning dynamics.
\end{abstract}

\newpage

\section{Introduction}
\label{sec:intro}

Estimating a population mean---such as disease prevalence, the proportion of households below the poverty line, or the fraction of files containing malware---is fundamental across scientific and policy applications. However, collecting ground-truth outcomes can be prohibitively expensive, often requiring human expert annotation, specialized testing, or additional surveys. In many modern settings, however, covariates and machine learning (ML) predictions are readily available for the \textit{entire population}, even when ground-truth outcomes can only be obtained for a small subset. How can we fully leverage this information to decide which outcomes to collect and efficiently estimate a population mean, with valid confidence intervals?

As a concrete example, consider a policymaker seeking to estimate the proportion of households living below the poverty line following a major economic disruption. Suppose demographic information is available for many households, but determining their current economic circumstances requires costly surveys. An ML model trained on historical data could predict each household's probability of poverty, although its predictions may be imperfect, particularly under changing economic conditions. Given a limited survey budget, the policymaker could survey households uniformly at random, obtaining a statistically valid sample mean but potentially leaving efficiency gains from ML predictions untapped. Alternatively, they could preferentially survey households whose poverty status the model is most uncertain about, update the model as new responses arrive, and adapt subsequent sampling accordingly. However, simply averaging the adaptively collected responses can introduce selection bias, while relying on ML predictions alone risks bias from prediction errors.

The challenge, therefore, is to leverage the full covariate pool and progressively acquired responses to construct efficient estimators and confidence intervals under adaptive, nonuniform sampling.

\paragraph{Problem setup} Formally, suppose we observe covariates ${x_1, \ldots, x_N}$ with corresponding unobserved responses ${y_1, \ldots, y_N}$, each of which may be collected at a (potentially substantial) cost. Given a labeling budget $n \ll N$, we seek to sequentially label data points, so as to construct an estimator and confidence interval, by leveraging, at each timestep, the full covariate pool, previously observed responses, and a black-box ML model that can be updated as new responses arrive. Specifically, we aim to use ML predictions both to adaptively select which responses to collect and to improve the resulting estimator through regression adjustment. Our primary goal is to construct an efficient estimator and statistically valid confidence interval for the finite-population mean 
\begin{align}
    \label{eqn:finite-pop-estimand}
    \theta := \frac{1}{N} \sum_{i=1}^N y_i,
\end{align}
treating the population as fixed and the sampling process alone as random. We also consider inference for the superpopulation mean $\mathbb{E}[Y]$ when the $(x_i,y_i)$ are assumed to have been drawn i.i.d.~from an underlying distribution.

\paragraph{Positioning} A natural starting point for the above problem is Active Statistical Inference (ASI) \citep{zrnic2024active}, which leverages ML predictions to guide label acquisition and construct statistically valid estimators and confidence intervals. In its sequential formulation (SAI), observations are processed in a fixed, predetermined order, with the ML model potentially updated as labels are acquired. Crucially, at round $t$, SAI can only make a Bernoulli decision to either acquire or skip the response $y_t$ associated with the current observation $x_t$. Once an observation is skipped, its response cannot be acquired later. When the full covariate pool is available in advance, however, why restrict each round to a label-skip decision about the current observation, rather than proactively selecting which observation to label next from the entire remaining pool?

Motivated by this question, our earlier work \citep{wu2026efficient} introduced Pro-Active Inference (PAI) in the context of LLM evaluation, combining proactive acquisition with online-updated ML predictions and regression adjustment. However, this initial formulation sampled with replacement, potentially revisiting previously labeled observations and wasting sampling opportunities. Extending PAI to sampling without replacement, however, requires more than modifying the acquisition policy: the estimator and its underlying martingale-based analysis must also be reformulated, as naive modifications can violate conditional unbiasedness and lead to severe confidence-interval miscoverage.

\paragraph{Our contributions} In the present work, we therefore develop a general, without-replacement formulation of PAI, together with a reformulated estimator and corresponding martingale-based inferential guarantees.

We establish its finite-sample unbiasedness and asymptotic validity, develop an improved finite-sample variance estimator, and extend PAI to superpopulation inference without requiring Bonferroni corrections or $\alpha$-splitting. Theoretically, we characterize PAI's potential efficiency gains over Batch Active Inference (BAI) in a quasi-oracle setting, highlighting the role of prediction-residual imbalance. Empirically, through experiments across eight diverse testbeds, we investigate PAI's statistical efficiency and confidence-interval coverage relative to existing methods. Beyond comparing performance, we develop interpretable diagnostics of residual imbalance, the relative difficulty of direct versus prediction-adjusted estimation, and the quality of online ML updates, together with targeted ablations, to understand when and why PAI gains or loses efficiency, and, more broadly, when ML-assisted statistical inference is beneficial.\footnote{All source code and analysis routines for reproducing all results in this paper can be found at
\url{https://github.com/skbwu/pro-active-inference}.}

Towards the completion of this manuscript, we also became aware of the closely related works of \citet{ashouritaklimi2026prediction,hamilton2026active}, which develop alternative approaches to adaptive sampling without replacement. While these approaches share important methodological features with PAI, they differ in their estimator constructions and uses of ML predictions. We discuss these and other related approaches in Section \ref{sec:related-work} and empirically compare PAI with the former in Section \ref{sec:results}.

\section{Related work}
\label{sec:related-work}

\paragraph{Prediction-powered inference} A natural starting point for estimating a population mean is uniform sampling, which yields statistically valid inference without involving an ML model, but may leave substantial efficiency gains from available ML predictions untapped. Prediction-Powered Inference (PPI; \citet{angelopoulos2023prediction}) leverages ML predictions to impute unobserved outcomes and corrects the resulting estimation error using a smaller sample of observed labels. Subsequent extensions include PPI++ \citep{angelopoulos2026ppi++}, which adapts the strength of the prediction-based adjustment to improve statistical efficiency, and cross-PPI \citep{zrnic2024cross}, which allows the available labeled data to be used for both training ML predictors and constructing inference through cross-fitting. While these approaches primarily consider nonadaptively acquired labels, PAI additionally uses ML predictions to guide which outcomes to acquire, potentially updating both its regression adjustment and acquisition policy as new labels are observed.

\paragraph{Active statistical inference} Our work builds on the Active Statistical Inference (ASI) framework of \citet{zrnic2024active}, which leverages ML predictions to guide label acquisition while constructing statistically valid estimators and confidence intervals for statistical estimands. In its sequential formulation (Sequential Active Inference; SAI), observations arrive in a stream, with covariates modeled as i.i.d.\ draws from a superpopulation, and the method makes Bernoulli label-skip decisions according to a potentially updated ML-based acquisition policy, subject to a prespecified labeling budget. Once an observation is skipped, however, its label cannot be acquired later. When the full covariate pool is available in advance, PAI instead proactively selects which observation to label next from the remaining unlabeled pool, using the full covariate information and previously acquired labels to update its ML-based regression adjustment and acquisition policy. Related developments include optimized and robust acquisition policies for ASI \citep{ma2026optimized,li2026robust}, as well as applications to human annotation and evaluation of LLM outputs \citep{gligoric2025can,angelopoulos2025cost}.

In our earlier work \citep{wu2026efficient}, we introduced a with-replacement formulation of PAI for LLM evaluation, using sampling policies and regression adjustments based on an online-updated working Bayesian factor model. However, this formulation allowed previously labeled observations to be sampled again, potentially wasting part of the labeling budget on repeated queries. The present work instead samples exclusively from the remaining unlabeled observations, ensuring that each query acquires a new label and that the full budget is spent on distinct observations, while developing the corresponding estimator and inferential guarantees. Finally, \citet{tolochinsky2026valid} also consider sampling without replacement in the context of best-model identification for LLMs, but use uniform sampling rather than an adaptive, ML-guided acquisition policy.

\paragraph{Active testing and measurement} PAI is also closely connected to a growing literature on efficient statistical inference in the presence of missing outcome data, particularly when labels are acquired adaptively. Most closely related is a line of work studying importance-weighted estimation under adaptive label acquisition \citep{farquhar2021statistical,kossen2021active,hamilton2026active,ashouritaklimi2026prediction}. With the exception of \citet{hamilton2026active}, these works focus primarily on estimating the risk of a fixed ML model, but their core methods readily apply to our population mean estimation task outlined in Section~\ref{sec:intro}. In particular, the active-testing literature \citep{farquhar2021statistical,kossen2021active} develops unbiased estimators under adaptive sampling without replacement.

Building on this literature, one of the methods most closely related to PAI is Active Measurement (AM; \citealp{hamilton2026active}), which studies the estimation of finite-population quantities, such as bird counts, rather than the prediction risk of a fixed ML model. AM uses online-updated ML predictions to guide acquisition and develops alternative weighting schemes and variance estimators for combining importance-weighted estimates. In particular, AM generalizes the estimation framework underlying \citet{farquhar2021statistical}'s LURE estimator by considering alternative schemes for aggregating conditionally unbiased per-round estimates. Indeed, at the level of the per-round estimator, AM coincides with PAI when our prediction-based regression adjustment is omitted. PAI instead incorporates online-updated ML predictions into both acquisition and regression adjustment, potentially improving sampling efficiency while preserving finite-sample unbiasedness.

Another closely related method is Prediction-Powered Active Testing (PPAT; \citealp{ashouritaklimi2026prediction}), which similarly builds on the active-testing literature, but additionally incorporates ML predictions directly into the estimator through a prediction-powered control variate. In all of their experiments, PPAT uses a fixed ML proxy for regression adjustment and a potentially updated surrogate model to guide acquisition, with a tunable control-variate coefficient. The authors further note that the control-variate proxy could also be updated as new labels are acquired, although they do not investigate this extension experimentally. In contrast to their experimental implementation, PAI directly incorporates online-updated ML predictions into its imputation and residual correction.

Like PAI, AM and PPAT construct their estimators by aggregating per-round importance-weighted estimates, which are conditionally unbiased under fixed tuning parameters. Under suitable regularity conditions, both approaches establish asymptotic Gaussian limiting distributions through martingale arguments, enabling the construction of asymptotically valid confidence intervals using corresponding variance estimators. However, PPAT aggregates its per-round estimates using nonuniform LURE-based weights \citep{farquhar2021statistical}, whereas PAI weights them equally.

Among these closely related approaches, we focus our empirical comparisons on PPAT, since its use of ML predictions in both acquisition and regression-adjusted estimation makes it a particularly direct methodological comparator to PAI.

\section{Pro-Active Inference methodology}
\label{sec:pai-methods}

We first describe PAI's adaptive sampling procedure, then develop estimators and confidence intervals for finite-population means, before extending our methodology to superpopulation inference.

\subsection{Sampling setup and notation}
\label{sec:sampling-setup-and-notation}

Let $\mathcal{I} = (1, \ldots, N)$ index the finite population, and let $O_{t-1} \subseteq \mathcal{I}$ denote the indices with observed responses before round $t$, with $O_0 = \emptyset$. Given a labeling budget $n \ll N$, at each round $t=1, \ldots, n$, PAI constructs an acquisition distribution (or \textit{sampling policy}) $q_t$ over the remaining unlabeled indices $R_{t-1} := \mathcal{I} \setminus O_{t-1}$, using the full covariate pool and all previously acquired responses. It then samples an index $I_t \sim q_t$, observes its response $y_{I_t}$, and updates $O_t := O_{t-1} \cup \{ I_t \}$. We denote by $\fhat^{(t-1)}$ the ML predictor available before round $t$, initialized at $\fhat^{(0)}$ and potentially updated after each newly acquired response. Formally, let $\F_{t-1} := \sigma(\{I_s, y_{I_s}\}_{s=1}^{t-1})$. We require $\fhat^{(t-1)}$ and $q_t$ to be $\F_{t-1}$-measurable: depending only on past queries/labels.

The sampling policy $q_t$ may depend on both the full covariate pool and previously acquired labels, but not on the unobserved responses. A natural strategy is to assign higher sampling probabilities to observations where the ML predictor is most uncertain. For binary outcomes, writing $p_t(x) := \fhat^{(t-1)}(x)$, common uncertainty scores include $u_t(x) = 2\min(p_t(x), 1-p_t(x))$ \citep{zrnic2024active} and $u_t(x) = \sqrt{p_t(x) (1 - p_t(x))}$, the latter motivated by Neyman allocation. Our experiments use the former, with uniform mixing to ensure positive sampling probabilities (Section \ref{sec:experimental-setup}), although our theoretical framework supports general adaptive sampling policies satisfying the regularity conditions below.

\subsection{Pro-Active Inference estimator for finite-population estimands}
\label{sec:finite-pop-pai-estimator}

After adaptively collecting $n$ labels without replacement, we propose the following Pro-Active Inference (PAI) estimator of the finite-population mean:
\begin{align}
\label{eqn:pai-estimator-for-mean}
    \thetahat_{n} := \frac{1}{n}\sum_{t=1}^n \phi_t, \quad \phi_t := \frac{1}{N}\left( \sum_{i \in O_{t-1}} y_i + \sum_{i \notin O_{t-1}} \hat{f}^{(t-1)}(x_{i}) \right) + \frac{1}{N}\frac{y_{I_t} - \hat{f}^{(t-1)}(x_{I_t})}{q_t(I_t)}.
\end{align}
Intuitively, the first term in $\phi_t$ represents our current plug-in estimate of $\theta$---using observed labels where available and imputing ML predictions for the remaining unobserved labels---while the second term provides an importance-weighted bias correction for prediction errors. In Theorem \ref{thm:pai-asymptotic-distribution}, we show that $\thetahat_n$ is unbiased for $\theta$ at every finite $n$ and derive its asymptotic distribution using a martingale central limit theorem under mild regularity conditions.

\begin{theorem}
\label{thm:pai-asymptotic-distribution}
    Consider a triangular array of problems (indexed by $m$) such that $n_m \uparrow \infty$, $N_m \uparrow \infty$, and $n_m/N_m \to \rho \in [0, 1 - \epsilon]$ for some $\epsilon > 0$.\footnote{We restrict the sampling fraction to be bounded away from $1$ to focus on the non-census regime. For readability, we omit the subscripts $m$ henceforth.} For each $t$, suppose that (a) $\fhat^{(t-1)}(\cdot)$ and $q_t(\cdot)$ are $\F_{t-1}$-measurable, and that (b) conditional on $\F_{t-1}$, $q_t(\cdot)$ is a probability distribution on $R_{t-1} := \mathcal{I} \setminus O_{t-1}$ with $q_t(i) \geq \tau / |R_{t-1}| = \tau / (N - t + 1)$ for all $i \in \R_{t-1}$, for some fixed $\tau \in (0, 1]$.

    Then, $\thetahat_n$ is unbiased for $\theta$ at every finite $n$. Furthermore, under mild regularity conditions (variance stabilization, Lindeberg, and uniform fourth-moment control for weighted residuals; Assumptions \ref{a1:variance_stabilization}-\ref{a3:4th-moment-control} in Appendix \ref{regularity-assumptions:pai-thm}), $\thetahat_n$ obeys the following limiting distribution:
    \begin{align}
        \sqrt{n}\left( \thetahat_n - \theta \right) / \sigmahat_{n} \toD \mathcal{N}\left( 0, 1\right).
    \end{align}
    And, the variance estimator is given by the following, where $\thetahat_{t-1} := \frac{1}{t-1} \sum_{s=1}^{t-1} \phi_s$ for $t \geq 2$:
    \begin{align}
        \label{eqn:og-var-estimator}
        \sigmahat_n^2 &:= \hat{A}_n - \Bhat_n\\
        &:= \frac{1}{nN^2} \sum_{t=1}^n \frac{\left( y_{I_t} - \hat{f}^{(t-1)}\left( x_{I_t} \right) \right)^2}{q_t(I_t)^2} - \frac{1}{nN^2} \sum_{t=2}^n \left( N\thetahat_{t-1} - \sum_{i \in O_{t-1}} y_i - \sum_{i \notin O_{t-1}} \hat{f}^{(t-1)}(x_i) \right)^2.
    \end{align}
\end{theorem}

\begin{proof}
    Please see Appendix \ref{proof:pai-thm}.
\end{proof}

By Theorem \ref{thm:pai-asymptotic-distribution}, it follows that an asymptotic $(1-\alpha)$-level confidence interval for $\theta$ is given by
\begin{align}
    \label{eqn:finite-pop-interval}
    \left[ \thetahat_n \pm z_{1-\alpha/2} \left( \sigmahat_n / \sqrt{n} \right) \right].
\end{align}

\subsection{An improved finite-sample variance estimator}
\label{sec:bias-corrected-variance-estimator}

Although $\sigmahat^2_n := \Ahat_n - \Bhat_n$ is consistent for the limiting variance $\sigma^2$ governing the martingale CLT in Theorem \ref{thm:pai-asymptotic-distribution} (see proof in Appendix \ref{proof:pai-thm}), empirically, it can occasionally underestimate variance at finite labeling budgets, yielding confidence intervals that are too narrow and coverage below the nominal $1-\alpha$ level.

Inspired by the higher-order bias-correction philosophy of the ABC and BCa bootstrap methods (c.f.~\cite{efron1987better,diciccio1992more}), we derive an additive correction to $\sigmahat_n^2$ using a second-order Taylor expansion. Specifically, we examine the finite-sample error introduced by substituting $\thetahat_{t-1}$ for the unknown population mean $\theta$ in our variance estimator (Equation \ref{eqn:og-var-estimator}). This motivates the following bias-corrected variance estimator $\sigmahat^2_{n,\text{cor}}$, where $\Ahat_n$ and $\Bhat_n$ are defined above in Theorem \ref{thm:pai-asymptotic-distribution}:
\begin{align*}
    \sigmahat^2_{n,\text{cor}} := \Ahat_n - \Bhat_n + \Chat_n, \quad \text{with } \Chat_n := \frac{1}{n}\sum_{s=1}^{n-1} \left( \frac{1}{N} \frac{y_{I_s} - \fhat^{(s-1)}(x_{I_s})}{q_s(I_s)} \right)^2 w_s \quad \text{and} \quad w_s := \sum_{k=s}^{n-1}\frac{1}{k^2}.
\end{align*}
Importantly, the correction $\Chat_n$ is nonnegative and vanishes asymptotically, preserving the asymptotic normality guarantee of Theorem \ref{thm:pai-asymptotic-distribution}. We provide the full derivation, proof of preserved asymptotic validity, and a characterization of when the correction reduces mean squared error in Appendix \ref{appendix:bias-corrected-variance-estimator}. We use the bias-corrected variance estimator $\sigmahat^2_{n,\text{cor}}$ in all our experiments and recommend it as the default for practitioners.

\subsection{Pro-Active Inference estimator for superpopulation estimands}
\label{sec:superpop-pai-estimator}

In the preceding subsections, we focused on the finite-population mean $\theta := N^{-1} \sum_{i=1}^N y_i$ as our estimand, treating $\{ (x_i, y_i )\}_{i=1}^N$ as a fixed finite population. 

However, oftentimes scientists may instead wish to adaptively sample from a finite population without replacement while making inference about the larger superpopulation from which the finite population was drawn. For example, following a natural disaster, policymakers may survey a subset of households from a representative sample of residents to assess their economic circumstances, while seeking to draw statistically principled conclusions about the broader population of affected households. In this subsection, we extend PAI to accommodate such superpopulation estimands.

Formally, we assume a superpopulation model with $(x_i, y_i) \simIID F$, with mean $\theta := \mathbb{E}[y_i]$ and variance $\tau^2 := \Var(y_i)$. Our goal is to construct a confidence interval for the superpopulation mean $\theta$. For this subsection only, define $\theta_N := N^{-1} \sum_{i=1}^N y_i$ to be the size-$N$ finite-population mean of $y_1, \ldots, y_N$, where the $y_i$ are now treated as random variables. 

Let $\thetahat_n$ be our PAI estimator for $\theta_N$ from Equation \ref{eqn:pai-estimator-for-mean}. The critical observation is that
\begin{align}
    \label{eqn:superpopulation-error-decomposition}
    \sqrt{n}\left( \thetahat_n - \theta \right) = \underbrace{\sqrt{n}(\thetahat_n - \theta_N)}_{\text{finite-population sampling error}} + \underbrace{\frac{\sqrt{n}}{\sqrt{N}} \sqrt{N}(\theta_N - \theta)}_{\text{superpopulation error}}.
\end{align}
Intuitively, this decomposition separates the uncertainty from observing only a subset of the finite population from the additional uncertainty arising because the finite population itself is a sample from a larger superpopulation. Using this decomposition, we obtain the following limiting distribution.

\begin{proposition}
    \label{prop:superpopulation-asymptotic-distribution}
    Retain the setup and assumptions of Theorem \ref{thm:pai-asymptotic-distribution}, where $\sigma^2$ denotes the limiting variance governing PAI's finite-population martingale CLT. Further suppose that $(x_i, y_i) \simIID F$ with mean $\theta := \mathbb{E}[y_i]$ and variance $\tau^2 := \Var(y_i)$. Then,
    \begin{align}
    \label{eqn:asymptotic-distribution-superpop}
        \sqrt{n}\left( \thetahat_n - \theta \right) \toD \mathcal{N}\left( 0, \sigma^2 + \rho \tau^2 \right),
    \end{align}
    where $\rho := \lim n/N$ from Theorem \ref{thm:pai-asymptotic-distribution}.
\end{proposition}

\begin{proof}
    Please see Appendix \ref{appendix:proof-of-superpopulation-asymptotic-distribution}.
\end{proof}
By Proposition \ref{prop:superpopulation-asymptotic-distribution} and Slutsky's Theorem, an asymptotically valid $(1-\alpha)$-confidence interval for the superpopulation mean $\theta$ is
\begin{align}
    \label{eqn:superpopulation-interval}
    \left[ \thetahat_n \pm z_{1 - \alpha / 2} \sqrt{\frac{\sigmahat_{n,\text{cor}}^2 + (n/N) \tauhat_n^2}{n}} \right],
\end{align}
where $\sigmahat_{n,\text{cor}}^2$ is our consistent bias-corrected variance estimator from Section \ref{sec:bias-corrected-variance-estimator}, and $\tauhat_n^2$ is any consistent estimator of $\tau^2$. To estimate the superpopulation variance $\tau^2 = \mathbb{E}[y_i^2] - \theta^2$, one simple approach is to apply PAI again to estimate the finite-population second moment $\theta_{N,2} := N^{-1} \sum_{i=1}^N y_i^2$, using the \textit{exact same queried indices} $I_1, \ldots, I_n$ used to construct $\thetahat_n$. This procedure simply treats $y_i^2$ as the outcome, with corresponding ML predictions (e.g., $\hat{g}^{(t)}(x_i) := ( \fhat^{(t)}(x_i) )^2$). Denoting the resulting PAI estimator by $\thetahat_{n,2}$, we construct
\begin{align}
    \label{eqn:second-pai-estimator-for-second-moment}
    \tauhat_n^2 := \thetahat_{n,2} - \thetahat_n^2,
\end{align}
which is consistent for $\tau^2$ under the corresponding regularity conditions of Theorem \ref{thm:pai-asymptotic-distribution}. Importantly, this procedure requires \textit{no additional labels}, since observing $y_{I_t}$ also reveals $y_{I_t}^2$. For binary outcomes $y_i \in \{ 0, 1 \}$, we can alternatively use the even simpler $\tauhat_n^2 = \thetahat_n (1 - \thetahat_n)$, since $y_i^2 = y_i$.

Furthermore, we note that our superpopulation confidence interval avoids both Bonferroni corrections and $\alpha$-splitting by directly accounting for both finite-population sampling uncertainty and superpopulation variability in a single asymptotic distribution.

\section{A theoretical characterization of PAI's efficiency}
\label{sec:anti-sai/bai}

Having established the PAI estimator and its inferential guarantees in Section \ref{sec:pai-methods}, we now theoretically investigate when and why PAI can yield efficiency gains over active inference in the finite-population setting.

For our theoretical analysis, we consider a quasi-oracle setting where both PAI and Batch Active Inference (BAI), the batch variant of SAI, share the same fixed ML model $f$, with oracle-known prediction residuals $e_i := y_i - f(x_i)$ and absolute residuals $r_i := |e_i|$.\footnote{Technically, we are working with BAI (just written using sequential notation) rather than SAI because the ML model is fixed. In other words, we are treating BAI as a special case of SAI.} Recall that BAI estimates the finite-population mean using
\begin{align}
    \label{eqn:bai-sai}
    \thetahat_{\text{active}} := \frac{1}{N} \sum_{t=1}^N \Gamma_t, \quad \Gamma_t = f(x_t) + \left( y_t - f(x_t) \right) \frac{\xi_t}{\pi_t(x_t)},
\end{align}
where $\pi_t(x_t)$ is the probability of labeling observation $t$ at time $t$ and $\xi_t \sim \text{Bern}(\pi_t(x_t))$ indicates whether its label was acquired. Following \citet[Section 4 and Equation 10]{zrnic2024active}'s recommended oracle sampling rule, we let both BAI and PAI sample proportionally to the true absolute prediction residuals:
\begin{align}
    \text{BAI:} \quad \pi_t(x_t) \propto r_t \implies \pi_t(x_t) = \frac{n r_t}{\sum_{i=1}^N r_i}, \quad \text{PAI:} \quad q_t(i) \propto r_i \implies q_t(i) = \frac{r_i}{\sum_{k \notin O_{t-1}} r_k},
\end{align}
with the $n$ present in BAI's label-skip probability because we want the expected number of BAI queries to be $n$.\footnote{Technically, we are assuming no clipping to $[0, 1]$ is necessary for mathematical simplicity. We also ignore uniform $\tau$-mixing and budget-stabilization (c.f. Sections 6 and 7.2 of \cite{zrnic2024active}) as we are in an idealized quasi-oracle regime.} By construction, PAI will achieve an exact budget of $n$ labels.

We adopt this quasi-oracle comparison setting to isolate the efficiency differences between PAI and BAI, giving both methods access to the same fixed ML model and analogous oracle sampling policies based on the true absolute prediction residuals. Fixing the ML model and omitting uniform mixing and budget stabilization also makes our theoretical comparison analytically tractable.

Because both PAI and BAI are unbiased for $\theta$, to characterize their performance differences, it suffices to compare their variances. Let $R_{t-1} = \mathcal{I} \setminus O_{t-1}$ represent the set of unlabeled indices after PAI sampling round $t-1$, where $\mathcal{I} = [1:N]$ is our full set of indices. Furthermore, for any set of indices $R \subseteq \mathcal{I}$, define the sums of unsigned and signed residuals $U(R)$ and $S(R)$, respectively:
\begin{align}
    \label{eqn:sums-of-unsigned/signed-residuals}
    U(R) := \sum_{i \in R} r_i, \quad S(R) := \sum_{i \in R} e_i.
\end{align}
Finally, define $b_t$ to be the \textit{net signed direction} of the residual mass before PAI sampling round $t$, and define $\bbar^2$ to be the \textit{weighted-average squared net signed direction} of the residual mass over the $n$ PAI sampling rounds:
\begin{align}
    b_t := \frac{S(R_{t-1})}{U(R_{t-1})} = \frac{\sum_{i \in R_{t-1}} \left( y_i - f(x_i) \right)}{\sum_{i \in R_{t-1}} |y_i - f(x_i)|}, \quad \bbar^2 := \frac{\sum_{t=1}^n \mathbb{E}\left[ U(R_{t-1})^2 b_t^2\right]}{\sum_{t=1}^n \mathbb{E}\left[ U(R_{t-1})^2 \right]} \in [0, 1].
\end{align}
Under the above setup, we obtain a theoretical characterization of PAI's sampling efficiency over BAI.

\begin{theorem}
\label{thm:pai-bai-var}
    Under the quasi-oracle setting above, assuming the below relative variance ratio is well-defined, we have
    \begin{align}
    \label{eqn:pai-bai-var-expression}
        \frac{\Var(\thetahat_{\text{active}})}{\Var(\thetahat_{\text{PAI}})} = \frac{1}{1 - \bbar^2} \cdot \frac{1}{1 + \eta_n},
    \end{align}
    where $\eta_n \geq 0$ is a remainder term capturing additional finite-budget and residual-magnitude effects (see Appendix \ref{appendix:proof-of-anti-bai-thm} for the full derivation). Furthermore, if $n r_{\text{max}} / U(R_0) \to 0$,\footnote{The sufficient condition above admits the decomposition $n r_{\text{max}} / U(R_0) = n/N \cdot r_{\text{max}} / \bar{r}$, where $\bar{r} := N^{-1} \sum_{i=1}^N r_i.$ Since $r_{\text{max}} \geq \bar{r}$, this condition necessarily requires $n/N \to 0$, which is permitted by Theorem \ref{thm:pai-asymptotic-distribution}. More precisely, the labeling fraction must vanish sufficiently quickly relative to $r_{\text{max}} / \bar{r}$. In particular, $n/N \to 0$ suffices if $r_{\text{max}} / \bar{r} = O(1)$.} with $r_{\text{max}} = \max_{i \in [1:N]} r_i$, then $\eta_n \to 0$, yielding
    \begin{align*}
        \frac{\Var(\thetahat_{\text{active}})}{\Var(\thetahat_{\text{PAI}})} = \frac{1 + o(1)}{1 - \bbar^2}.
    \end{align*}
\end{theorem}
\begin{proof}
    Please see Appendix \ref{appendix:proof-of-anti-bai-thm}.
\end{proof}
In context, Theorem \ref{thm:pai-bai-var} suggests that PAI can substantially outperform BAI when (i) the ML model's prediction residuals systematically point in the same direction, rather than canceling each other out (i.e., when $\bbar^2$ is large); and (ii) the sampling policy effectively prioritizes observations with large absolute prediction residuals (as in our quasi-oracle setting). We do emphasize that such prediction residual imbalance is \textit{not} a measure of predictive accuracy: a highly accurate model can still have a large $\bbar^2$ if its small residuals consistently point in one direction. Alternatively, a large $\bbar^2$ could arise when the ML model predicts most observations accurately, but exhibits a systematic bias on a small subgroup of the population. 

In Section \ref{sec:results}, we will investigate whether this theoretical insight empirically carries over to more realistic settings with learned acquisition policies and online ML updates.

\section{Experimental setup}
\label{sec:experimental-setup}

Our empirical experiments seek to characterize PAI's sampling efficiency and statistical coverage relative to its intellectual predecessor Sequential Active Inference (SAI) \citep{zrnic2024active}, the closely related contemporary method Prediction-Powered Active Testing (PPAT) \citep{ashouritaklimi2026prediction}, and the simpler baselines Prediction-Powered Inference (PPI) \citep{angelopoulos2023prediction} and uniform sampling. We evaluate these methods across eight real-world datasets spanning computational biology, LLM evaluations, economics, fraud detection, and cybersecurity.

We focus on finite-population mean estimation with binary outcomes $y_i \in \{ 0, 1\}$. The finite-population setting allows us to directly evaluate statistical coverage against the known population mean, whereas the superpopulation mean is generally unknown for our testbeds. Focusing on binary outcomes also provides simple, well-established uncertainty scores for guiding acquisition (e.g., $u_t(x) = 2\min(f(x), 1-f(x))$, where $f$ is a probabilistic ML predictor; see \citet[Section 4]{zrnic2024active}), avoiding additional modeling and acquisition-policy choices that arise in regression settings and could complicate the interpretation of performance differences. We emphasize that none of these methods are restricted to binary outcomes.

Importantly, to isolate differences between the inferential methods, we use the same underlying ML model and initialization across PAI, SAI, PPI, and PPAT, as well as the same online learning procedure wherever applicable. We also use the same underlying uncertainty scores and uniform-mixing parameter for PAI and SAI, following the recommendations of \citet{zrnic2024active}, while respecting PPAT's own estimator and recommended acquisition mechanisms. Rather than engaging in extensive dataset-specific hyperparameter tuning, our goal is to understand how effectively these methods leverage the same underlying predictive information for statistical inference.

\subsection{Datasets and ML models} 

All dataset processing routines can be found in the \texttt{data\_processing} notebook in our Github \href{https://github.com/skbwu/pro-active-inference}{repository}.

\paragraph{ACS} We use the ACSIncome task \citep{ding2021retiring}, where the binary outcome indicates whether an individual's annual income exceeds $\$50,000$, using demographic covariates. Adapting their distribution-shift setting, we train the initial ML model on California and evaluate on Puerto Rico, motivated by the importance of label-efficient data collection in resource-constrained settings, such as following natural disasters. Our test population consists of $9,071$ individuals with $11$ covariates.

\paragraph{AlphaFold} We use the protein-disorder dataset from \citet{angelopoulos2023prediction}, originally derived from \citet{bludau2022structural}. Each observation corresponds to an amino-acid residue, with a binary outcome indicating whether it belongs to an intrinsically disordered region. We estimate the proportion of disordered residues among $10,802$ observations, using probabilistic ML predictions obtained from a logistic regression model fitted to AlphaFold-derived structural features. Since this dataset provides only the predictions and labels, without the underlying covariates, all ML-assisted methods use fixed predictions without online updates. For SAI, this corresponds to the Batch Active Inference (BAI) setting.

\paragraph{LLM evaluations} We study estimation of LLM accuracy on two benchmark suites: MMLU-Pro \citep{wang2024mmlu}, containing $12,032$ questions, and a composite of BBH, GPQA, IFEval, MATH, and MuSR \citep{suzgun2023challenging,rein2024gpqa,zhou2023instruction,hendrycksmath2021,spraguemusr}, containing $9,574$ questions. Following \citet{wu2026efficient}, we use their train/test splits and $64$-dimensional question and model factors fitted to historical evaluation data from \citet{open-llm-leaderboard-v2,myrzakhan2024open}, with approximately $2.2$K held-out LLMs per suite. For each held-out LLM, we fix the historical question factors and initialize its model factor at the empirical mean of historical model factors, using their empirical marginal variances to initialize the diagonal curvature for online logistic regression updates. These testbeds also allow us to directly compare the sampling efficiency of our without-replacement PAI estimator against its with-replacement predecessor from \citet{wu2026efficient}.

\paragraph{BODMAS} We use the BODMAS malware detection dataset \citep{bodmas}, where each observation represents a software sample with a binary label indicating whether it is malicious or benign. We employ a chronological 50/50 train-test split, fitting the initial ML model on the earlier observations and performing inference on the later $67,218$ observations. After removing features that are identically zero in the training split, each observation has $2,330$ covariates.

\paragraph{HIV} We use the HIV molecular activity dataset from MoleculeNet \citep{wu2018moleculenet}, accessed through \href{https://github.com/deepchem/deepchem}{DeepChem}. Each observation represents a chemical compound screened for its ability to inhibit HIV replication, with confirmed or moderately active compounds labeled as positive and inactive compounds labeled as negative. We represent compounds using ECFP molecular fingerprints and employ a 50/50 train-test split, yielding a test population of $20,560$ compounds with $1,025$ covariates.

\paragraph{Pew} We use Wave 168 of the Pew Research Center's \href{https://www.pewresearch.org/dataset/american-trends-panel-wave-168/}{American Trends Panel}, which surveys U.S. adults about online scams and related topics. Our binary outcome indicates whether a respondent reports having ever lost money because of an online scam or attack. We use demographic covariates and responses to other scam-related survey questions and employ a 50/50 train-test split, yielding a test population of $4,273$ respondents with $53$ covariates.

\paragraph{ULB} We use the ULB Credit Card Fraud Detection dataset \citep{dal2015calibrating}, containing anonymized credit card transactions from European cardholders, with a binary outcome indicating whether a transaction is fraudulent. The covariates comprise $28$ PCA-transformed features, transaction time, and amount. We employ a 50/50 train-test split, yielding a test population of $142,404$ transactions with $30$ covariates. This highly imbalanced dataset, with approximately $0.17\%$ fraudulent transactions, allows us to investigate statistical inference for rare events.

\paragraph{ML models and online updates} For all ML-assisted methods (i.e., PAI, SAI, PPI, and PPAT), we use the same underlying probabilistic ML model and online updating procedure whenever applicable. For all datasets except the LLM evaluation suites and AlphaFold, we standardize covariates using training-split statistics and fit an initial logistic regression on the training data, initializing its diagonal curvature approximation to the identity. For the two LLM evaluation suites, following \cite{wu2026efficient}, we instead initialize the model coefficients using the empirical mean of historical model factors and the diagonal curvature using their inverse empirical marginal variances, retaining the original question-factor representations. For all datasets except AlphaFold, we update the logistic regression model after every acquired label using computationally-efficient diagonal stochastic Newton updates adapted from \citet{bercu2020efficient}, as detailed in Appendix \ref{appendix:olr-updates}. AlphaFold is the sole exception, since only fixed ML predictions are available.

\subsection{Baselines and competitor methods}

\paragraph{PAI and SAI} For PAI and SAI, we use the same underlying ML predictor and online updating procedure. Following \citet[Sections 4 and 7.2]{zrnic2024active}, both methods also use the uncertainty score $u_t(x) = 2 \min(\fhat^{(t-1)}(x), 1 - \fhat^{(t-1)}(x))$ and uniform-mixing parameter $\tau = 0.5$. At each round $t$, PAI samples without replacement from the remaining unlabeled pool $R_{t-1}$ according to the acquisition policy
\begin{align}
    \label{eqn:pai-standardized-sampling-policy}
    q_t(i) = \frac{\tau}{|R_{t-1}|} + (1 - \tau) \frac{u_t(x_i)}{\sum_{j \in R_{t-1}} u_t(x_j)}.
\end{align}
In contrast, SAI processes the test observations sequentially, at round $t$ choosing whether to acquire label $y_t$ with probability $\pi_t(x_t)$. Its mean estimator is
\begin{align}
    \label{eqn:sai}
    \thetahat_{\text{active}} := \frac{1}{N} \sum_{t=1}^N \left( \fhat^{(t-1)}(x_t) + \left( y_t - \fhat^{(t-1)}(x_t) \right) \frac{\xi_t}{\pi_t(x_t)} \right),
\end{align}
where $\xi_t \sim \text{Bern}(\pi_t(x_t))$ indicates whether the label is acquired. To target a labeling budget of $n$ observations from a test population of size $N$, we use the following nominal acquisition probability:
\begin{align}
    \pi_t^{\text{nom}}(x_t) = \tau \left( \frac{n}{N} \right) + \left( 1 - \tau \right) \frac{n u_t(x_t)}{\sum_{j=1}^N u_t(x_j)}.
\end{align}
To prevent budget underutilization and ensure good-faith comparisons, we additionally implement the SAI budget-stabilization rule of \citet[Equation 8]{zrnic2024active} and their periodic budget catch-up heuristic from Section 8.4. These adjustments are applied to the uncertainty-based probabilities before mixing with the uniform component, with the resulting acquisition probabilities $\pi_t(x_t)$ used in the SAI estimator.

\paragraph{PPI and uniform sampling} We implement PPI as SAI (Equation \ref{eqn:sai}) with uniform Bernoulli acquisition probabilities $\pi_t(x_t) = n/N$, retaining the shared ML predictions and online updates in the estimator. Our implementation of uniform sampling, while preserving the SAI estimator structure, additionally removes the ML predictions, recovering the corresponding inverse-probability-weighted estimator. Thus, PPI isolates the benefits of ML-assisted inference without adaptive acquisition, while uniform sampling serves as a baseline without any ML assistance. Because both methods use fixed acquisition probabilities, neither requires SAI's budget-stabilization rule or periodic budget catch-up heuristic.

\paragraph{PPAT} We implement the online-tuned $\hat{\lambda}_M$ variant of PPAT from \citet{ashouritaklimi2026prediction}, using its prediction-powered control-variate estimator and proposed acquisition policy specialized to binary outcomes (Section 3.2). Following Section 5.1, we initialize $\hat\lambda_M = 0.5$ and update it after every $100$ acquired labels. We also use uniform mixing with $\tau = 0.5$, matching PAI and SAI for consistency. We initialize PPAT's control-variate proxy and acquisition surrogate using the same initial ML predictor as the other ML-assisted methods. During sampling, we update the acquisition surrogate after every acquired label using our shared online learning procedure, while keeping the control-variate proxy fixed. Importantly, this is precisely the fixed-proxy configuration adopted by the PPAT authors themselves when benchmarking against Active Statistical Inference in Appendix C.4.3, combined with their online surrogate updates in Section 5.1. Finally, we select the online-tuned $\hat{\lambda}_M$ variant over the fixed choices $\lambda \in \{ 0.5, 1 \}$ because it had the strongest performance in our preliminary comparisons.

\subsection{Budgets, metrics, and coverage}

\paragraph{Budgets and coverage} For each dataset, we consider ten evenly-spaced budgets $n \in [2.5\%, 25\%]$ of the test population and target $95\%$ nominal coverage for all experiments. All methods construct nominal $95\%$ Wald-type confidence intervals of the form $[\thetahat_{\text{method}}(n) \pm 1.96 \sigmahat_{\text{method}}(n)]$, using their respective plug-in variance estimators. For each method and budget, we run $1000$ random seeds on the non-LLM datasets, and $100$ random seeds for each of approximately $2.2$K held-out LLMs on the two LLM evaluation suites. We compute empirical coverage as the proportion of confidence intervals containing the true finite-population mean $\theta$, averaging over random seeds and, for the LLM evaluation suites, held-out LLMs.

\paragraph{Metrics} To quantify sampling efficiency, we use effective sample size multipliers (ESS multipliers) based on estimated variances. For a method $M$ and reference method $B$, both evaluated at labeling budget $n$, we define the \textit{ESS multiplier of $M$ with respect to $B$} as
\begin{align}
    \gamma_{\text{eff}}^{M/B}(n) := \hat\sigma_B^2(n) / \hat\sigma_M^2(n).
\end{align}
In particular, taking uniform sampling $U$ as the reference, the \textit{effective sample size} (ESS) of method $M$ is
\begin{align}
    n_{\text{eff}} := \gamma_{\text{eff}}^{M/U}(n) \cdot n,
\end{align}
with uniform sampling having ESS $n$ by definition. An ESS multiplier greater than $1$ indicates higher sampling efficiency than the reference method. For example, at budget $n$, a multiplier of $1.5$ with respect to uniform sampling means achieving approximately the same confidence interval width as uniform sampling with $1.5n$ labels, corresponding to roughly $33\%$ fewer labels for comparable precision. We report ESS multipliers relative to uniform sampling, SAI, PPI, and PPAT to facilitate comparisons across methods.

\section{Results}
\label{sec:results}

Our empirical results seek to investigate four primary questions. First, how does PAI compare in sampling efficiency and statistical coverage against its predecessors: uniform sampling, PPI, and SAI? Second, to what extent does our theoretical characterization of PAI's efficiency over active inference (Theorem \ref{thm:pai-bai-var}) carry over to real-data settings with learned acquisition policies and online ML updates? Third, how does PAI compare against the closely related contemporary method PPAT? Fourth, and perhaps most importantly, beyond simply reporting wins and losses, can we identify interpretable explanations for when and why PAI outperforms or underperforms these competing methods, yielding insights applicable to ML-assisted adaptive statistical inference more broadly?

We investigate each of these questions below. Before proceeding, however, we first discuss an important coverage failure on the ULB Credit Card Fraud dataset.

\paragraph{Coverages on the ULB dataset} We observe that all three methods using ML predictions to guide acquisition---PAI, SAI, and PPAT---severely undercover the nominal $95\%$ level at the smallest labeling budgets. At the $2.5\%$ budget, empirical coverages are $64.6\%$, $82.0\%$, and $63.4\%$ for PAI, SAI, and PPAT, respectively, compared to $92.1\%$ for both PPI and Uniform.

We hypothesize that this undercoverage is driven by the rarity of positive labels, potentially exacerbated by confidently mispredicted positives under ML-guided acquisition. Of the $N=142,404$ observations in the ULB test set, only $N_1 = 223$ (approximately $0.15\%$) are fraudulent transactions ($y_i = 1$). At the $2.5\%$ labeling budget, even uniform sampling would acquire only approximately $5.6$ positive labels in expectation, potentially compromising the accuracy of normal-approximation confidence intervals. Of these $223$ positive observations, the shared initial ML model assigns $39$ of them a fraud probability below $1\%$. These confidently mispredicted positives have large prediction residuals, yet can receive low acquisition probabilities under the ML-guided sampling policies. 

\begin{figure}[!t]
    \centering
    \includegraphics[width=1.0\linewidth]{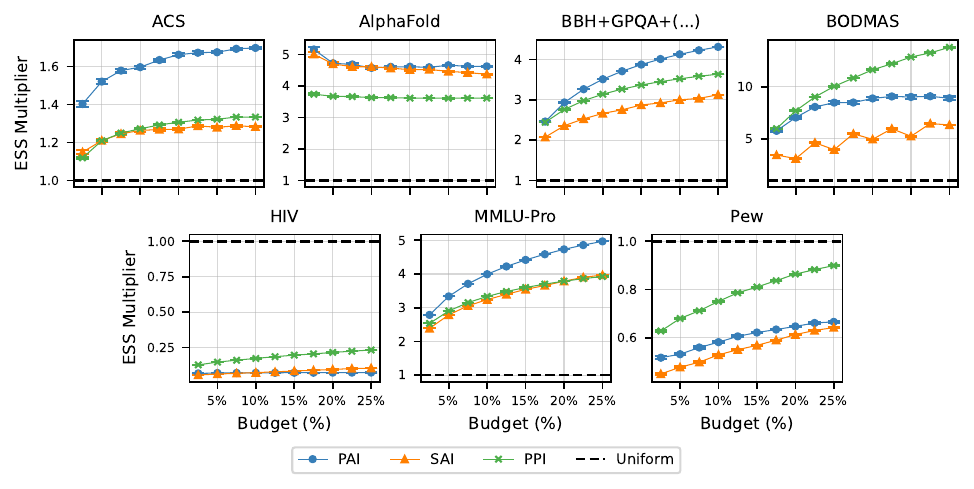}
    \caption{\textbf{Comparing sampling efficiency of Pro-Active Inference (PAI), Sequential Active Inference (SAI), and Prediction-Powered Inference (PPI) with respect to uniform sampling across seven datasets.} Each subplot shows the effective sample size (ESS) multiplier with respect to uniform sampling as a function of labeling budget for PAI, SAI, and PPI. For example, on ACS at a $10\%$ budget, an ESS multiplier of $1.6$ means PAI matches the confidence interval width of uniform sampling at a $16\%$ budget. Standard errors are shown, but negligible on the plot scale. \textbf{Higher ESS multipliers indicate better sampling efficiency, with values above $1$ indicating an improvement over uniform sampling.}}
    \label{fig:ess-pai-sai-ppi-uniform}
\end{figure}

Crucially, observing one of these confidently mispredicted positives can provide both (i) a substantial upward correction to the estimated population mean $\thetahat$; and (ii) evidence of substantial sampling variability, particularly through inverse-probability weighting contributions to the variance estimators. Sampling trajectories that fail to acquire these observations can therefore underestimate both $\thetahat$ and its standard error $\sigmahat$, yielding overly narrow confidence intervals that fail to cover the true population mean. 

Although the rarity of positive outcomes does not itself invalidate asymptotic inference, we suspect that the normal approximations underlying these methods are unreliable at the labeling budgets considered, with ML-guided acquisition potentially compounding this problem by undersampling rare but consequential positive observations.

Because apparent gains in sampling efficiency are not meaningful without adequate statistical coverage, we exclude ULB from the subsequent ESS comparisons and focus on the remaining seven datasets.

\subsection{How does PAI compare against uniform sampling, PPI, and SAI?}
\label{sec:pai-unif-ppi-sai}

On the remaining seven datasets, all methods achieve near-nominal coverage (Figure \ref{fig:all-coverages}), enabling meaningful comparisons of sampling efficiency. First, PAI achieves higher ESS than SAI on all $10$ of $10$ tested budgets on ACS, BBH+GPQA+IFEval+MATH+MuSR, BODMAS, MMLU-Pro, and Pew; on $9$ of $10$ tested budgets on AlphaFold; and on $3$ of $10$ tested budgets on HIV (Figure \ref{fig:ess-pai-sai-ppi-uniform}). One of the motivations for developing without-replacement PAI was to address limitations of our earlier with-replacement PAI formulation encountered in LLM evaluations \citep{wu2026efficient}. Reassuringly, on both LLM evaluation suites, without-replacement PAI consistently outperforms its with-replacement predecessor in sampling efficiency (Appendix Figure \ref{fig:pai-with-vs-without-replacement}). \textbf{In short, PAI represents a compelling improvement over SAI in pool-based settings}, supporting its role as a natural successor to SAI when the full covariate bank is available for proactive sampling.

\textbf{However, this does not mean that more sophisticated adaptive inference necessarily outperforms simpler methods such as Uniform and PPI.} Somewhat surprisingly, PAI only outperforms uniform sampling on $5$ of $7$ datasets (albeit across all tested budgets), while considerably underperforming on HIV and Pew at all tested budgets (Figure \ref{fig:ess-pai-sai-ppi-uniform}). Against PPI, PAI outperforms on $4$ of $7$ datasets (again, across all tested budgets), while substantially underperforming on BODMAS, HIV, and Pew. We defer a detailed investigation of these performance trends to Section \ref{section:diagnostics-and-why}.

\subsection{Does our theoretical efficiency characterization hold empirically?}
\label{sec:anti-bai-empirics}

\begin{figure}
    \centering
    \includegraphics[width=1.0\linewidth]{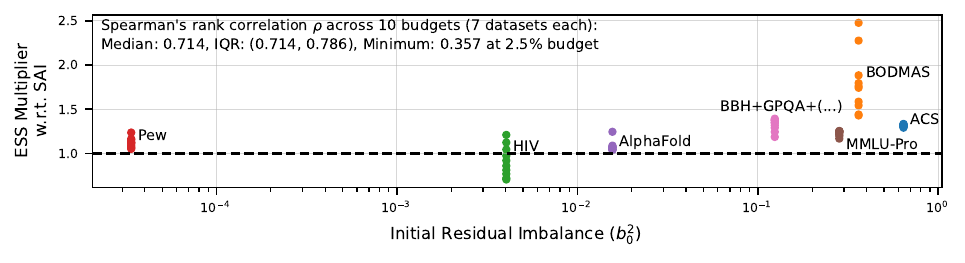}
    \caption{\textbf{ESS multipliers of PAI over SAI as a function of initial residual imbalance $b_0^2$.} Each color represents one dataset, with ten points corresponding to the ten labeling budgets tested. Unlike Figure \ref{fig:ess-pai-sai-ppi-uniform}, ESS multipliers here are measured with respect to SAI rather than uniform sampling. For example, an ESS multiplier of $2.5$ means that PAI matches the confidence interval width of SAI while using only $40\%$ as many labels. \textbf{Higher ESS multipliers indicate better sampling efficiency, with values above $1$ indicating that PAI outperforms SAI.} Spearman rank correlations were computed across the seven datasets at each fixed budget, and we report the median, interquartile range (IQR), and minimum of the resulting ten correlations.}
    \label{fig:sai-theory-corroboration}
\end{figure}

Motivated by Theorem \ref{thm:pai-bai-var}, we define the quasi-oracle \textbf{initial residual imbalance} (IRB) diagnostic $b_0^2$ as a one-step analogue of the $\bbar^2$ quantity in our theoretical characterization, which can be computed without running full sample trajectories:
\begin{align}
    \label{eqn:irb}
    b_0^2 := \left( \frac{\sum_{i=1}^N \left( y_i - \fhat^{(0)}(x_i) \right)}{\sum_{i=1}^N \left| y_i - \fhat^{(0)}(x_i) \right|} \right)^2.
\end{align}
Intuitively, IRB measures the extent to which the initial ML model's prediction residuals systematically point in the same direction, rather than canceling each other out. Values near $0$ indicate that positive and negative residuals largely cancel, while values near $1$ indicate that residuals of one sign dominate in aggregate magnitude. As discussed in Section \ref{sec:anti-sai/bai}, IRB is not necessarily indicative of predictive accuracy. 

Figure \ref{fig:sai-theory-corroboration} provides strong empirical corroboration of the qualitative prediction from Theorem \ref{thm:pai-bai-var}. Even with online-updated ML models and acquisition policies, IRB exhibits a strong positive rank correlation with PAI's ESS multiplier gains over SAI across labeling budgets. While we cannot expect the exact $1 / (1 - \bbar^2)$ functional form to carry over to these more realistic settings, the empirical relationship is striking: larger initial residual imbalance is consistently associated with greater efficiency gains for PAI over SAI.

\paragraph{On quasi-oracle diagnostics} We emphasize that the diagnostics in this section (i.e., IRB, ECB, and IVR in Equations \ref{eqn:irb}-\ref{eqn:ivr}, respectively) are quasi-oracle, computed using the full test population of covariates and labels. These labels are used only for retrospective diagnostic analysis, not for guiding the estimators' acquisition policies or constructing their estimates. Our goal is to identify interpretable properties associated with sampling efficiency, without relying on post-hoc analyses of the realized sampling trajectories.

\subsection{How does PAI compare against PPAT?}
\label{sec:pai-vs-ppat}

\begin{figure}
    \centering
    \includegraphics[width=1.0\linewidth]{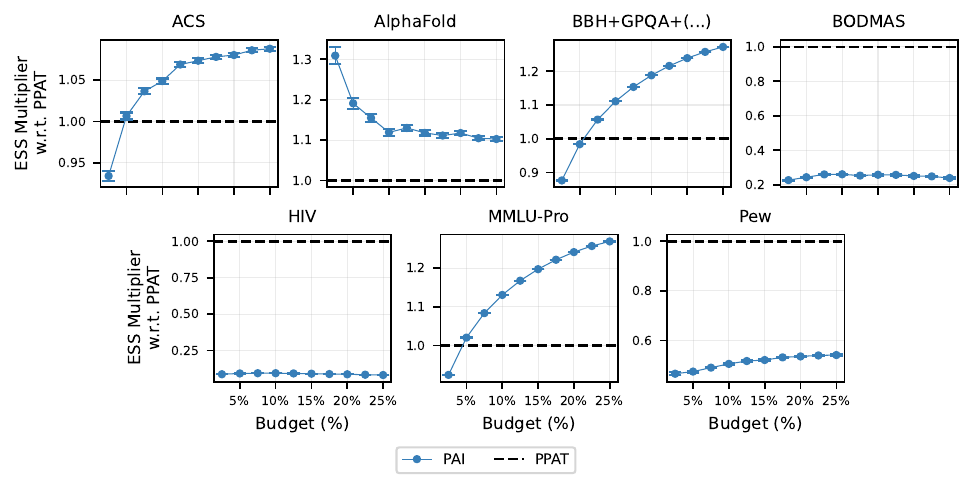}
    \caption{\textbf{ESS multipliers of PAI over PPAT across seven datasets.} Each subplot shows PAI's ESS multiplier with respect to PPAT as a function of labeling budget. For example, on AlphaFold at a $2.5\%$ budget, an ESS multiplier of $1.3$ means PAI matches the confidence interval width of PPAT at a $3.25\%$ budget. Standard errors are shown but are generally negligible at the plot scale. \textbf{Higher ESS multipliers indicate better sampling efficiency, with values above $1$ indicating that PAI outperforms PPAT.}}
    \label{fig:pai-vs-ppat}
\end{figure}
PAI outperforms PPAT on $10$ of $10$ tested budgets on AlphaFold; on $9$ of $10$ budgets each on ACS and MMLU-Pro; and on $8$ of $10$ budgets on BBH+GPQA+IFEval+MATH+MuSR (Figure \ref{fig:pai-vs-ppat}). However, PAI considerably underperforms PPAT on all tested budgets on BODMAS, HIV, and Pew. Curiously, these are precisely the same three datasets where PAI also underperforms PPI (Section \ref{sec:pai-unif-ppi-sai}). This common pattern motivates our investigation into when and why PAI gains or loses efficiency relative to competing methods, which we pursue in the following subsection.

\subsection{What can we learn from PAI's efficiency gains and losses?}
\label{section:diagnostics-and-why}

\begin{figure}[!t]
    \centering
    \includegraphics[width=1.0\linewidth]{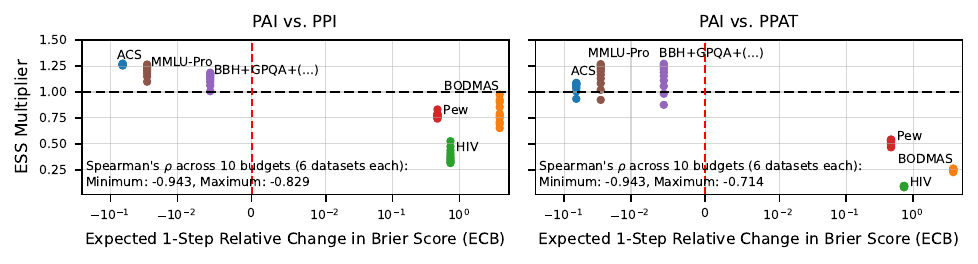}
    \caption{\textbf{ESS multipliers of PAI over PPI and PPAT as a function of the expected $1$-step relative change in Brier score (ECB).} Each color represents one dataset, with ten points corresponding to the ten labeling budgets tested. The left and right subplots show ESS multipliers with respect to PPI and PPAT, respectively. For example, on ACS in the left subplot, an ESS multiplier of $1.25$ means that PAI matches the confidence interval width of PPI while using $20\%$ fewer labels. \textbf{Higher ESS multipliers indicate better sampling efficiency, with values above $1$ indicating that PAI outperforms the corresponding competitor.} Spearman rank correlations were computed across the six datasets at each fixed budget, and we report the minimum and maximum of the resulting ten correlations in each subplot.}
    \label{fig:ppat+ppi-diagnostic}
\end{figure}

\subsubsection{When and why does PAI outperform or underperform PPI and PPAT?}
\label{section:diagnostics-ppi+ppat}

One potential explanation for PAI's relative efficiency gains and losses against PPI and PPAT lies in its increased sensitivity to the quality of online ML updates. Specifically, PAI incorporates updated ML predictions into \textit{both} (i) its plug-in estimate and residual correction; and (ii) its acquisition policy. In contrast, PPI, implemented in our experiments as SAI with a uniform acquisition policy, is only affected by online ML updates through its prediction-powered regression adjustment. Similarly, PPAT, as implemented in our experiments, only uses the online-updated ML model to guide acquisition (which \citet{ashouritaklimi2026prediction} refer to as the ``surrogate"), while the prediction-based control variate retains the initial ML predictions (which the authors refer to as the ``proxy"). Importantly, this fixed-proxy formulation is also used by the PPAT authors themselves in their empirical comparisons against Active Statistical Inference \citep[Section 5.1 and Appendix C.4.3]{ashouritaklimi2026prediction}.

Thus, while PAI has two channels of sensitivity to online prediction quality, PPI and PPAT, as implemented in our experiments, each have only one. We hypothesize that deteriorating online ML predictions can harm PAI's sampling efficiency through both channels: less accurate imputations can increase residual-estimation variability, while poorly-guided acquisition can further inflate variance by assigning unfavorable sampling probabilities to observations with large prediction residuals. Conversely, when online ML updates improve predictions, PAI can potentially benefit from both improved imputation and better-guided acquisition.

To investigate this hypothesis, we quantify how the initial online ML update affects probabilistic prediction quality. For each dataset, we define the \textbf{expected $1$-step relative change in Brier score} ($\text{ECB}$) as follows:
\begin{align}
    \text{ECB} := \frac{\sum_{i=1}^N q_1(i) \text{Brier}_{-i}\left( \fhat^{(0,+i)} \right)}{\sum_{i=1}^N q_1(i) \text{Brier}_{-i}\left( \fhat^{(0)} \right)} - 1, \quad \text{where } \text{Brier}_{-i}(f) := \frac{1}{N-1} \sum_{j=1,j\neq i}^N \left( y_j - f(x_j) \right)^2.
    \label{eqn:ecb}
\end{align}
Here, $q_1(i)$ is PAI's initial sampling policy (Equation \ref{eqn:pai-standardized-sampling-policy}), $\text{Brier}_{-i}(f)$ is the Brier score of model $f$ evaluated over the test population excluding index $i$, and $\fhat^{(0,+i)}$ is the online-updated ML model obtained after observing $(x_i, y_i)$ (see Equation \ref{eqn:olr-updates} in Appendix \ref{appendix:olr-updates}). Intuitively, ECB measures the relative change in our ML model's expected probabilistic prediction error after one online update, averaging over the initial PAI sampling policy. Negative ECB values indicate model \textit{improvement}, while positive values indicate \textit{deterioration}.

\FloatBarrier

Figure \ref{fig:ppat+ppi-diagnostic} provides strong empirical support for our hypothesis. Among the six datasets with online ML updates,\footnote{We exclude AlphaFold from this diagnostic because it provides fixed probabilistic ML predictions without access to an updatable ML model.} all three datasets where PAI outperforms both PPI and PPAT on the majority of tested budgets (ACS, MMLU-Pro, and BBH+GPQA+IFEval+MATH+MuSR) correspond to negative ECB values, while all three datasets where PAI underperforms both methods (BODMAS, HIV, and Pew) correspond to positive ECB values. Beyond this striking separation, ECB is also strongly negatively rank-correlated with PAI's ESS multipliers over both competing methods across the tested labeling budgets.

To further investigate this hypothesis, we conduct a frozen-ML ablation across all six datasets, holding PAI's initial ML model fixed throughout sampling and thereby disabling online updates in both imputation and acquisition (Appendix Figure \ref{fig:frozen-pai-ablation-all-datasets}). Strikingly, on the three datasets with negative ECB values and where PAI outperformed both PPI and PPAT (ACS and the two LLM evaluation suites), freezing substantially reduces PAI's relative sampling efficiency, reversing its advantages over both PPI and PPAT. Conversely, on the three datasets with positive ECB values (BODMAS, HIV, and Pew), freezing substantially improves PAI's relative efficiency, reversing its disadvantages against PPI on all three datasets and against PPAT on Pew and $6$ of $10$ budgets on HIV. On BODMAS, freezing substantially narrows PAI's disadvantage against PPAT, although it does not eliminate it.

We emphasize that the competing methods' implementations remain unchanged throughout these ablations, as our goal is to understand the effect of online updating on PAI's relative efficiency, using the competing methods as fixed reference points, rather than attempting to optimize PAI's performance against them. Furthermore, because freezing simultaneously affects PAI's imputation and acquisition, these ablations identify the combined effect of online updating without isolating the contribution of either channel.

\subsubsection{When and why does PAI outperform or underperform uniform sampling?}
\label{sec:pai-vs-uniform}

\begin{figure}[!t]
    \centering
    \includegraphics[width=1.0\linewidth]{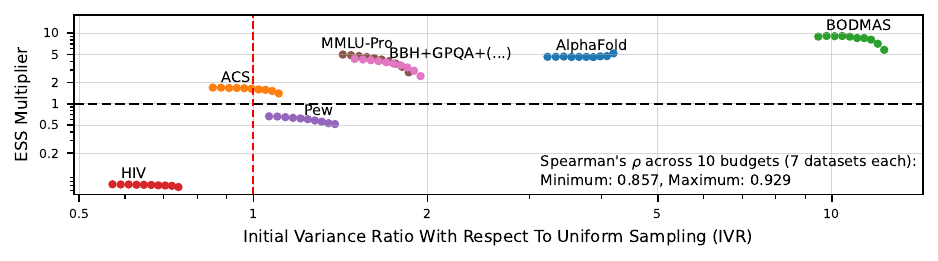}
    \caption{\textbf{ESS multipliers of PAI over uniform sampling as a function of initial variance ratio (IVR).} Each color represents one dataset, with ten points corresponding to the ten labeling budgets tested. An $\text{IVR}(n)$ greater than $1$ indicates that the initial variance approximation favors PAI over uniform sampling, suggesting a potential efficiency advantage for PAI. \textbf{Higher ESS multipliers indicate better sampling efficiency, with values above $1$ indicating that PAI outperforms uniform sampling.} Spearman rank correlations were computed across the seven datasets at each fixed budget, and we report the minimum and maximum of the resulting ten correlations.}
    \label{fig:uniform-diagnostic}
\end{figure}

Our empirical results raise an important question: when is it advantageous to use ML predictions and adaptive acquisition to estimate a population mean, rather than simply estimating the mean directly through uniform sampling? Intuitively, PAI uses ML predictions to reduce the estimation problem to estimating a residual correction. However, this residual-estimation problem is not necessarily less variable than directly estimating the population mean, particularly if the acquisition policy assigns unfavorable sampling probabilities to observations with large prediction residuals.

Interestingly, the two datasets where PAI underperforms uniform sampling (HIV and Pew) constitute a strict subset of the three datasets where PAI underperforms both PPI and PPAT (BODMAS, HIV, and Pew).

To investigate when residual estimation under PAI is more favorable than direct estimation under uniform sampling, we define the \textbf{initial variance ratio} (IVR) at labeling budget $n$ as 
\begin{align}
    \label{eqn:ivr}
    \text{IVR}(n) := \frac{\Var(\thetahat_{\text{unif},n})}{n^{-1} \Var(\phi_1)} = \frac{\left( 1 - n/N \right)\theta}{\sum_{i=1}^N q_1(i) \left( \bar{f}^{(0)} + \frac{1}{N} \frac{y_i - \fhat^{(0)}(x_i)}{q_1(i)} - \theta \right)^2}, \quad \text{where } \bar{f}^{(0)} := \frac{1}{N} \sum_{i=1}^N \fhat^{(0)}(x_i).
\end{align}
Here, uniform sampling is implemented using the no-ML Bernoulli label-skip estimator of \citet{zrnic2024active}, with uniform inclusion probability $n/N$. The second equality in Equation \ref{eqn:ivr} holds for binary outcomes $y_i \in \{ 0, 1\}$ under this sampling design.

Intuitively, IVR compares the variance of directly estimating the \textit{population mean} under uniform sampling against the initial variance of estimating the \textit{ML model's mean residual} under PAI's acquisition policy. Larger IVR values suggest that the residual-estimation problem is less noisy than directly estimating the outcome mean, and thus greater potential efficiency gains for PAI over uniform sampling. To obtain a simple, interpretable diagnostic without relying on realized sampling trajectories, we approximate PAI's variance at budget $n$ using its initial one-step variance, scaled by $1/n$. Thus, IVR captures the joint effect of the initial ML predictions and acquisition policy, without accounting for subsequent online model updates, corresponding changes in the acquisition policy, or without-replacement effects.

Figure \ref{fig:uniform-diagnostic} shows a strong positive rank correlation between IVR and PAI's realized ESS multipliers over uniform sampling across labeling budgets. This suggests that even the initial ML predictions and acquisition policy already contain substantial information about PAI's eventual relative sampling efficiency, despite the subsequent complexities of online ML updates and adaptive sampling.

\paragraph{What did we learn?} Tying together our insights from both Sections \ref{section:diagnostics-ppi+ppat} and \ref{sec:pai-vs-uniform}, we note that IVR and ECB offer complementary perspectives on PAI's sampling efficiency. While IVR characterizes the relative difficulty of the initial residual-estimation problem under PAI's acquisition policy, ECB characterizes how the quality of ML predictions changes after an initial online update. Our frozen-ML PAI ablations (described in Section \ref{section:diagnostics-ppi+ppat}; results in Appendix Figure \ref{fig:frozen-pai-ablation-all-datasets}) further demonstrate that subsequent online learning can substantially improve or deteriorate PAI's realized efficiency relative to its frozen-model counterpart. Thus, understanding when ML-assisted adaptive inference is advantageous requires considering both the initial difficulty of the residual-estimation problem and the subsequent behavior of the online learner.

\section{Discussion and conclusion}

\paragraph{Contributions} In this work, we develop a more general, without-replacement formulation of Pro-Active Inference (PAI), building on our earlier formulation for efficient LLM evaluation \citep{wu2026efficient}, and provide a comprehensive statistical treatment of its properties and broader applicability to ML-assisted statistical inference. Our original goal was to build upon the Active Statistical Inference framework of \citet{zrnic2024active} while pursuing two key desiderata: (i) making fuller use of the available covariate information by proactively selecting observations for labeling from the entire remaining pool, rather than restricting each round to a label-skip decision about the current observation in a predetermined sequence; and (ii) sampling without replacement.

Methodologically, we establish PAI's finite-sample unbiasedness and asymptotic validity, develop an improved finite-sample variance estimator, and extend PAI to superpopulation inference without the need for Bonferroni corrections or $\alpha$-splitting. We further characterize PAI's potential efficiency gains over Active Statistical Inference through a theoretical analysis highlighting the role of prediction-residual imbalance. Empirically, through extensive experiments across eight datasets, interpretable theory-motivated diagnostics, and frozen-ML ablations, we investigated when and why PAI outperforms or underperforms competing inference methods, with the broader goal of understanding when ML-assisted adaptive statistical inference itself is advantageous.

\paragraph{On our original goal of building upon Active Statistical Inference} Turning to our original goal of building upon Active Statistical Inference, we believe PAI largely achieves both of our stated desiderata. Beyond satisfying them by construction, PAI achieves greater sampling efficiency than SAI across the vast majority of our tested settings, while using the same underlying ML predictor and uncertainty score. Unlike SAI, which processes observations in a predetermined order and cannot revisit skipped observations, PAI can reconsider all remaining unlabeled observations after each acquisition. This flexibility can be particularly valuable when certain observations have large prediction residuals, as PAI retains opportunities to acquire their labels even if they were not selected initially. Our theoretical characterization further identifies prediction-residual imbalance as a key condition under which PAI can achieve substantial efficiency gains over Batch Active Inference. Reassuringly, our empirical initial residual imbalance (IRB) diagnostic strongly corroborates this theoretical intuition, even in more realistic settings with online ML updates and learned acquisition policies.

\paragraph{Understanding the efficiency of ML-assisted adaptive inference} Beyond PAI's improvements over SAI, our experiments highlight that more sophisticated ML-assisted adaptive inference does not necessarily outperform simpler methods such as PPI and uniform sampling. Our proposed diagnostics offer two complementary perspectives on when and why these efficiency gains or losses occur: (i) IVR characterizes how much easier it is to estimate the ML model's prediction residuals rather than the population mean directly; while (ii) ECB characterizes whether online ML updates improve or deteriorate prediction quality. Our frozen-ML ablations further corroborate the importance of online learning dynamics, while our comparisons with PPAT highlight how different uses of updated ML predictions in acquisition and estimation can affect efficiency.

\paragraph{From retrospective diagnostics to prospective guidance} Although our proposed diagnostics IRB, ECB, and IVR were computed retrospectively using the fully labeled test population, an appealing feature is that they summarize important characteristics of the inference problem into a small number of interpretable statistics, without requiring full sampling trajectories. This raises the possibility of approximating these diagnostics using representative validation data or small pilot samples, or reasoning about their values using domain-specific knowledge. For example, IVR depends partly on the unknown population mean $\theta$, for which domain scientists may be able to provide plausible ranges, allowing sensitivity analyses when combined with validation- or pilot-based estimates of the relevant residual moments. Similarly, preliminary assessments of online ML updates could help inform whether to continually update or freeze the predictor, or perhaps consider a different update rule. Such assessments could eventually guide the choice of inference method, acquisition policy, or online updating strategy before committing to expensive label collection. Developing and validating these prospective diagnostics constitutes promising future work, particularly given their dependence on the quality and representativeness of the available auxiliary data.

\paragraph{Other limitations and future directions} Our experiments deliberately focused on binary outcomes, standardized the underlying ML predictors, and, for PAI and SAI, the acquisition heuristics, to facilitate interpretable comparisons. However, a growing literature investigates optimized and robust acquisition policies for Active Statistical Inference \citep{ma2026optimized,li2026robust}, and alternative ML models and online updating procedures could yield different method-specific relative efficiencies. Rather than extensively tuning each method's acquisition policy and ML model, our goal was to understand how effectively the competing inferential procedures leverage comparable predictive information. Exploring alternative, potentially domain-specific acquisition policies and online learning strategies for PAI therefore constitutes an important direction for future work.

Furthermore, the severe undercoverage observed on the ULB dataset highlights the challenges of ML-assisted adaptive inference for rare events at finite labeling budgets. Other promising directions include extending our empirical investigations to continuous outcomes and superpopulation settings, developing PAI for more general estimands such as M-estimation, and exploring further domain-specific applications.

More broadly, we hope this work provides a foundation for developing statistically reliable, label-efficient inference methods that make effective use of the increasingly abundant ML predictions and covariate information available in modern applications.

\subsubsection*{Acknowledgments}

We thank Kevin Luo, Alexandru Lopotenco, Etaash Katiyar, Valérie Ho, Noah Cowan, and Ginnie Ma for their many insightful comments, objective outside advice, and helpful discussions. S.W. and Y.N. were both supported by the National Science Foundation Graduate Research Fellowship Program. E.J.C. was supported by the Office of Naval Research grant N000014-20-1-2157. We also gratefully acknowledge support from Arena Intelligence's Academic Partnerships Program. This work used Marlowe \citep{marlowe2025}, Stanford University's GPU-based Computational Instrument, supported by Stanford HAI and Stanford Research Computing. We thank Stanford University and the Stanford Research Computing Center for providing the computational resources and support that contributed to these research results. 

\bibliography{references}
\bibliographystyle{tmlr}

\newpage
\appendix

\section{Proofs and other formulas}

\subsection{Regularity assumptions for Theorem \ref{thm:pai-asymptotic-distribution}}
\label{regularity-assumptions:pai-thm}

We enumerate some mild regularity conditions for Theorem \ref{thm:pai-asymptotic-distribution}: Assumptions \ref{a1:variance_stabilization}-\ref{a2:lindeberg_condition} are standard conditions for martingale central limit theorems (c.f.~\citet[Corollary 3.2]{hall2014martingale}) while Assumption \ref{a3:4th-moment-control} is mild and used to facilitate consistency of our variance estimator. Define $\Delta_t := \phi_t - \theta$.

\begin{assumption}[Variance stabilization]
\label{a1:variance_stabilization}
    The average conditional variance converges in probability to a positive deterministic limit $\sigma^2$, which does not depend on $m$:
    \begin{align*}
        \frac{1}{n} \sum_{t=1}^{n} \Var\left( \Delta_t \mid \F_{t-1} \right) \toP \sigma^2.
    \end{align*}
\end{assumption}

\begin{assumption}[Lindeberg condition]
\label{a2:lindeberg_condition}
    The triangular array of problem instances in conjunction with the sampling policies $q_t(\cdot)$ obey the Lindeberg condition. Specifically, for every $\epsilon > 0$,
    \begin{align*}
        \frac{1}{n} \sum_{t=1}^n \mathbb{E}\left[ \Delta^2_t \bdone\left( |\Delta_t| > \epsilon \sqrt{n} \right) \mid \F_{t-1} \right] \toP 0.
    \end{align*}
\end{assumption}

\begin{assumption}[Uniform fourth-moment control for weighted residuals]
\label{a3:4th-moment-control}
    There exists a deterministic constant $C < \infty$ such that
    \begin{align*}
        \sup_{m \geq 1} \sup_{1 \leq t \leq n_m} \mathbb{E}\left[ \left( \frac{1}{N} \frac{y_{I_t} - \fhat^{(t-1)}(x_{I_t})}{q_t(I_t)} \right)^4 \right] \leq C.
    \end{align*}
\end{assumption}

\subsection{Proof of Theorem \ref{thm:pai-asymptotic-distribution}}
\label{proof:pai-thm}

\begin{proof}
    We begin by showing that the $\Delta_t$ are martingale differences---that their cumulative sums form a martingale. By definition of $\F_{t-1}$ and $q_t(\cdot) := P(I_t = \cdot \mid \F_{t-1})$, with division by $q_t(\cdot)$ well-defined by assumption,
    \begin{align*}
        \mathbb{E}\left[ \Delta_t \mid \F_{t-1} \right] = \mathbb{E}\left[ \frac{1}{N}\left( \sum_{i \in O_{t-1}} y_i + \sum_{i \notin O_{t-1}} \hat{f}^{(t-1)}(x_{i}) \right) + \frac{1}{N}\frac{y_{I_t} - \hat{f}^{(t-1)}(x_{I_t})}{q_t(I_t)} - \theta \mid \F_{t-1}\right]\\
        = \frac{1}{N} \sum_{i \in O_{t-1}} y_i + \frac{1}{N}\left( \sum_{i \notin O_{t-1}} \hat{f}^{(t-1)}(x_{i}) \right) + \mathbb{E}\left[\frac{1}{N}\frac{y_{I_t} - \hat{f}^{(t-1)}(x_{I_t})}{q_t(I_t)} \mid \F_{t-1} \right] - \theta \\
        = \frac{1}{N} \sum_{i \in O_{t-1}} y_i + \frac{1}{N}\left( \sum_{i \notin O_{t-1}} \hat{f}^{(t-1)}(x_{i}) \right) + \frac{1}{N} \sum_{i \notin O_{t-1}} \left( \frac{y_i - \fhat^{(t-1)}(x_i)}{q_t(i)} \right) q_t(i) - \theta\\
        = \frac{1}{N} \sum_{i \in \mathcal{I}} y_i - \theta = 0,
    \end{align*}
    which implies that $(\Delta_t, \F_t)$ is a martingale difference array. Now, it immediately follows that
    \begin{align*}
        \mathbb{E}\left[ \thetahat_n - \theta \right] = \mathbb{E}\left[ \frac{1}{n}\sum_{t=1}^n \phi_t - \theta \right] = \mathbb{E}\left[ \frac{1}{n}\sum_{t=1}^n \left( \phi_t - \theta \right) \right] = \mathbb{E}\left[ \frac{1}{n}\sum_{t=1}^n \Delta_t \right] = 0,
    \end{align*}
    implying that $\thetahat_n$ is unbiased for $\theta$ at every finite $n$. Now, by Assumptions \ref{a1:variance_stabilization}-\ref{a2:lindeberg_condition}, the martingale central limit theorem as presented in Corollary 3.2 of \cite{hall2014martingale} tells us that
    \begin{align*}
        \sqrt{n}\left( \thetahat_n - \theta \right) \toD \mathcal{N}\left( 0, \sigma^2 \right).
    \end{align*}
    To compute a consistent estimator for $\sigma^2$, we proceed by computing the martingale increment conditional variances. By definition of $\F_{t-1}$, we have
    \begin{align*}
        \Var\left( \Delta_t \mid \F_{t-1} \right)\\
        = \Var\left( \frac{1}{N}\left( \sum_{i \in O_{t-1}} y_i + \sum_{i \notin O_{t-1}} \hat{f}^{(t-1)}(x_{i}) \right) + \frac{1}{N}\frac{y_{I_t} - \hat{f}^{(t-1)}(x_{I_t})}{q_t(I_t)} - \theta \mid \F_{t-1} \right)\\
        = \Var\left( \frac{1}{N}\frac{y_{I_t} - \hat{f}^{(t-1)}(x_{I_t})}{q_t(I_t)} \mid \F_{t-1} \right)\\ = \mathbb{E}\left[ \left( \frac{1}{N}\frac{y_{I_t} - \hat{f}^{(t-1)}(x_{I_t})}{q_t(I_t)} \right)^2 \mid \F_{t-1} \right] - \left(  \mathbb{E}\left[ \frac{1}{N}\frac{y_{I_t} - \hat{f}^{(t-1)}(x_{I_t})}{q_t(I_t)} \mid \F_{t-1} \right] \right)^2\\
        = \frac{1}{N^2} \sum_{i \notin O_{t-1}} \frac{\left( y_i - \fhat^{(t-1)}(x_i) \right)^2}{q_t(i)} -  \frac{1}{N^2} \left( \sum_{i \notin O_{t-1}} \left( y_i - \fhat^{(t-1)}(x_i) \right) \right)^2.
    \end{align*}
    By Assumption \ref{a1:variance_stabilization} and substitution, we know that
    \begin{align*}
        \frac{1}{n} \sum_{t=1}^n \Var(\Delta_t \mid \F_{t-1}) = \frac{1}{n} \sum_{t=1}^n \left( \frac{1}{N^2} \sum_{i \notin O_{t-1}} \frac{\left( y_i - \fhat^{(t-1)}(x_i) \right)^2}{q_t(i)} -  \frac{1}{N^2} \left( \sum_{i \notin O_{t-1}} \left( y_i - \fhat^{(t-1)}(x_i) \right) \right)^2 \right) \\
        = \frac{1}{n N^2} \sum_{t=1}^n \left( \sum_{i \notin O_{t-1}} \frac{\left( y_i - \fhat^{(t-1)}(x_i) \right)^2}{q_t(i)} \right) - \frac{1}{n N^2} \sum_{t=1}^n \left( \sum_{i \notin O_{t-1}} \left( y_i - \fhat^{(t-1)}(x_i) \right) \right)^2\\
        := A_n - B_n \toP \sigma^2.
    \end{align*}
    To show that $\sigmahat_n^2 \toP \sigma^2$, it suffices to show that (a) $\hat{A}_n - A_n = o_p(1)$ and (b) $\hat{B}_n - B_n = o_p(1)$. Beginning with the first statement,
    \begin{align*}
        \hat{A}_n - A_n = \frac{1}{nN^2} \sum_{t=1}^n \frac{\left( y_{I_t} - \hat{f}^{(t-1)}\left( x_{I_t} \right) \right)^2}{q_t(I_t)^2} - \frac{1}{n N^2} \sum_{t=1}^n \left( \sum_{i \notin O_{t-1}} \frac{\left( y_i - \fhat^{(t-1)}(x_i) \right)^2}{q_t(i)} \right)\\
        = \frac{1}{n} \sum_{t=1}^n \frac{1}{N^2}\left( \frac{\left( y_{I_t} - \hat{f}^{(t-1)}\left( x_{I_t} \right) \right)^2}{q_t(I_t)^2} - \sum_{i \notin O_{t-1}} \frac{\left( y_i - \fhat^{(t-1)}(x_i) \right)^2}{q_t(i)} \right) := \frac{1}{n} \sum_{t=1}^n C_t. 
    \end{align*}
    Because we have
    \begin{align*}
        \mathbb{E}[C_t \mid \F_{t-1}] = \mathbb{E}\left[ \frac{1}{N^2}\left( \frac{\left( y_{I_t} - \hat{f}^{(t-1)}\left( x_{I_t} \right) \right)^2}{q_t(I_t)^2} - \sum_{i \notin O_{t-1}} \frac{\left( y_i - \fhat^{(t-1)}(x_i) \right)^2}{q_t(i)} \right) \mid \F_{t-1} \right]\\
        = \mathbb{E}\left[ \frac{1}{N^2}\left( \frac{\left( y_{I_t} - \hat{f}^{(t-1)}\left( x_{I_t} \right) \right)^2}{q_t(I_t)^2} \right) \mid \F_{t-1} \right] - \frac{1}{N^2} \left( \sum_{i \notin O_{t-1}} \frac{\left( y_i - \fhat^{(t-1)}(x_i) \right)^2}{q_t(i)} \right)\\
        = \sum_{i \notin O_{t-1}} \frac{1}{N^2}\left( \frac{\left( y_{i} - \hat{f}^{(t-1)}\left( x_{i} \right) \right)^2}{q_t(i)^2} \right) q_t(i) - \frac{1}{N^2} \left( \sum_{i \notin O_{t-1}} \frac{\left( y_i - \fhat^{(t-1)}(x_i) \right)^2}{q_t(i)} \right) = 0,
    \end{align*}
    we can deduce that $(C_t, \F_t)$ is a martingale difference array. By orthogonality of martingale increments (i.e., tower property), we further find that
    \begin{align*}
        \mathbb{E}\left[ \left( \hat{A}_n - A_n \right)^2 \right] = \mathbb{E}\left[ \left( \frac{1}{n} \sum_{t=1}^n C_t \right)^2 \right] = \frac{1}{n^2} \sum_{t=1}^n \mathbb{E}\left[ C_t^2 \right].
    \end{align*}
    Now, because $(a-b)^2 \leq 2a^2 + 2b^2$,
    \begin{align*}
        \mathbb{E}\left[ C_t^2 \right] \leq \mathbb{E}\left[ 2 \left( \frac{1}{N^2} \frac{\left( y_{I_t} - \hat{f}^{(t-1)}\left( x_{I_t} \right) \right)^2}{q_t(I_t)^2} \right)^2 + 2 \left( \frac{1}{N^2} \sum_{i \notin O_{t-1}} \frac{\left( y_i - \fhat^{(t-1)}(x_i) \right)^2}{q_t(i)}\right)^2\right]\\
        = 2 \mathbb{E}\left[ \left( \frac{1}{N} \frac{\left( y_{I_t} - \hat{f}^{(t-1)}\left( x_{I_t} \right) \right)}{q_t(I_t)} \right)^4\right] + 2 \mathbb{E}\left[ \left( \mathbb{E}\left[ \frac{1}{N^2} \frac{\left( y_{I_t} - \fhat^{(t-1)}(x_{I_t}) \right)^2}{q_t({I_t})^2} \mid \F_{t-1} \right]\right)^2 \right]\\
        = 2 \mathbb{E}\left[ \left( \frac{1}{N} \frac{\left( y_{I_t} - \hat{f}^{(t-1)}\left( x_{I_t} \right) \right)}{q_t(I_t)} \right)^4\right] + 2 \mathbb{E}\left[ \mathbb{E}\left[ \left( \frac{1}{N} \frac{\left( y_{I_t} - \fhat^{(t-1)}(x_{I_t}) \right)}{q_t({I_t})} \right)^4 \mid \F_{t-1} \right] \right]\\
        = 4 \mathbb{E}\left[ \left( \frac{1}{N} \frac{\left( y_{I_t} - \hat{f}^{(t-1)}\left( x_{I_t} \right) \right)}{q_t(I_t)} \right)^4\right] \leq 4C,
    \end{align*}
    with the last step by Assumption \ref{a3:4th-moment-control}. As such, we deduce that $\mathbb{E}\left[ \left( \hat{A}_n - A_n \right)^2 \right] \to 0$, which implies that $\hat{A}_n - A_n = o_p(1)$. Working on the second term, we first observe that
    \begin{align*}
        B_n = \frac{1}{n N^2} \sum_{t=1}^n \left( \sum_{i \notin O_{t-1}} \left( y_i - \fhat^{(t-1)}(x_i) \right) \right)^2 = \frac{1}{n N^2} \sum_{t=1}^n \left( N\theta - \sum_{i \in O_{t-1}} y_i - \sum_{i \notin O_{t-1}} \fhat^{(t-1)}(x_i)  \right)^2\\
        = \frac{1}{n N^2} \sum_{t=2}^n \left( N\theta - \sum_{i \in O_{t-1}} y_i - \sum_{i \notin O_{t-1}} \fhat^{(t-1)}(x_i)  \right)^2 + \frac{1}{nN^2} \left( N\theta - \sum_{i=1}^N \fhat^{(t-1)}(x_i) \right)^2.
    \end{align*}
    With this alternative form, and using the equality that $a^2 - b^2 = (a - b)^2 + 2b(a-b)$ to combine similar terms, we have the following:
    \begin{align*}
        \Bhat_n - B_n = \frac{1}{nN^2} \sum_{t=2}^n \left( N\thetahat_{t-1} - \sum_{i \in O_{t-1}} y_i - \sum_{i \notin O_{t-1}} \hat{f}^{(t-1)}(x_i) \right)^2\\ - \frac{1}{n N^2} \sum_{t=2}^n \left( N\theta - \sum_{i \in O_{t-1}} y_i - \sum_{i \notin O_{t-1}} \fhat^{(t-1)}(x_i)  \right)^2 - \frac{1}{nN^2} \left( N\theta - \sum_{i=1}^N \fhat^{(0)}(x_i) \right)^2\\
        = \frac{1}{nN^2} \sum_{t=2}^n \left[ \left( N\thetahat_{t-1} - \sum_{i \in O_{t-1}} y_i - \sum_{i \notin O_{t-1}} \hat{f}^{(t-1)}(x_i) \right)^2 - \left( N\theta - \sum_{i \in O_{t-1}} y_i - \sum_{i \notin O_{t-1}} \fhat^{(t-1)}(x_i)  \right)^2 \right]\\ - \frac{1}{nN^2} \left( N\theta - \sum_{i=1}^N \fhat^{(0)}(x_i) \right)^2\\
        = \frac{1}{n} \sum_{t=2}^n \left( \thetahat_{t-1} - \theta \right)^2 + \frac{2}{n} \sum_{t=2}^n \left( \theta - \frac{\sum_{i \in O_{t-1}} y_i + \sum_{i \notin O_{t-1}} \fhat^{(t-1)}(x_i)}{N} \right) \left( \thetahat_{t-1} - \theta \right)\\ - \frac{1}{n} \left( \theta - \frac{\sum_{i=1}^N \fhat^{(0)}(x_i)}{N} \right)^2.
    \end{align*}
    Taking absolute values, we have
    \begin{align*}
        |\Bhat_n - B_n| \leq \underbrace{\frac{1}{n} \sum_{t=2}^n \left( \thetahat_{t-1} - \theta \right)^2}_{T_1} + \underbrace{\frac{2}{n} \sum_{t=2}^n \left| \theta - \frac{\sum_{i \in O_{t-1}} y_i + \sum_{i \notin O_{t-1}} \fhat^{(t-1)}(x_i)}{N} \right| \left| \thetahat_{t-1} - \theta \right|}_{T_2}\\ + \underbrace{\frac{1}{n} \left( \theta - \frac{\sum_{i=1}^N \fhat^{(0)}(x_i)}{N} \right)^2}_{T_3}.
    \end{align*}
    Working on $T_1$, we have the following because the $\Delta_s$ are martingale differences and by orthogonality of martingale increments:
    \begin{align*}
        \thetahat_{t-1} - \theta = \frac{1}{t-1} \sum_{s=1}^{t-1} \Delta_s \implies \mathbb{E}\left[ \left(  \thetahat_{t-1} - \theta \right)^2\right] = \frac{1}{(t-1)^2} \sum_{s=1}^{t-1} \mathbb{E}\left[ \Delta_s^2 \right].
    \end{align*}
    Now, by Jensen's Inequality, the fact that $(a - b)^2 \leq 2a^2 + 2b^2$, and the following fact:
    \begin{align*}
        \mathbb{E}\left[ \frac{1}{N}\frac{y_{I_t} - \hat{f}^{(t-1)}(x_{I_t})}{q_t(I_t)} \mid \F_{t-1} \right] = \frac{1}{N} \sum_{i \notin O_{t-1}} \left( y_i - \fhat^{(t-1)}(x_i) \right),
    \end{align*}
    we have the following, noting that $\theta - \frac{1}{N} \sum_{i \in O_{s-1}} y_i = \frac{1}{N} \sum_{i \notin O_{s-1}} y_i$:
    \begin{align*}
        \mathbb{E}\left[ \Delta_s^2 \right] = \mathbb{E}\left[ \left( \frac{1}{N}\left( \sum_{i \in O_{s-1}} y_i + \sum_{i \notin O_{s-1}} \hat{f}^{(s-1)}(x_{i}) \right) + \frac{1}{N}\frac{y_{I_s} - \hat{f}^{(s-1)}(x_{I_s})}{q_s(I_s)} - \theta \right)^2 \right]\\
        \leq 2\mathbb{E}\left[ \left( \frac{1}{N}\frac{y_{I_s} - \hat{f}^{(s-1)}(x_{I_s})}{q_s(I_s)} \right)^2 \right] + 2\mathbb{E}\left[ \left( \theta - \frac{1}{N}\left( \sum_{i \in O_{s-1}} y_i + \sum_{i \notin O_{s-1}} \hat{f}^{(s-1)}(x_{i}) \right) \right)^2 \right]\\
        \leq 2\left(\mathbb{E}\left[ \left( \frac{1}{N}\frac{y_{I_s} - \hat{f}^{(s-1)}(x_{I_s})}{q_s(I_s)} \right)^4 \right]\right)^{1/2} + 2\mathbb{E}\left[ \left( \frac{1}{N} \sum_{i \notin O_{s-1}} y_i - \frac{1}{N} \sum_{i \notin O_{s-1}} \hat{f}^{(s-1)}(x_{i}) \right)^2 \right]\\
        = 2\left(\mathbb{E}\left[ \left( \frac{1}{N}\frac{y_{I_s} - \hat{f}^{(s-1)}(x_{I_s})}{q_s(I_s)} \right)^4 \right]\right)^{1/2} + 2\mathbb{E}\left[ \left(  \mathbb{E}\left[ \frac{1}{N}\frac{y_{I_s} - \hat{f}^{(s-1)}(x_{I_s})}{q_s(I_s)} \mid \F_{s-1} \right] \right)^2 \right]\\
        \leq 2\left(\mathbb{E}\left[ \left( \frac{1}{N}\frac{y_{I_s} - \hat{f}^{(s-1)}(x_{I_s})}{q_s(I_s)} \right)^4 \right]\right)^{1/2} + 2\mathbb{E}\left[ \mathbb{E}\left[ \left( \frac{1}{N}\frac{y_{I_s} - \hat{f}^{(s-1)}(x_{I_s})}{q_s(I_s)} \right)^2 \mid \F_{s-1} \right] \right]\\
        \leq 2\left(\mathbb{E}\left[ \left( \frac{1}{N}\frac{y_{I_s} - \hat{f}^{(s-1)}(x_{I_s})}{q_s(I_s)} \right)^4 \right]\right)^{1/2} + 2\mathbb{E}\left[ \left( \frac{1}{N}\frac{y_{I_s} - \hat{f}^{(s-1)}(x_{I_s})}{q_s(I_s)} \right)^2 \right]\\
        \leq 2\left(\mathbb{E}\left[ \left( \frac{1}{N}\frac{y_{I_s} - \hat{f}^{(s-1)}(x_{I_s})}{q_s(I_s)} \right)^4 \right]\right)^{1/2} + 2\left(\mathbb{E}\left[ \left( \frac{1}{N}\frac{y_{I_s} - \hat{f}^{(s-1)}(x_{I_s})}{q_s(I_s)} \right)^4 \right] \right)^{1/2} \leq 4C^{1/2}.
    \end{align*}
    Thus, we conclude that
    \begin{align*}
        \mathbb{E}[T_1] = \frac{1}{n}\sum_{t=2}^n \mathbb{E}\left[ \left( \thetahat_{t-1} - \theta \right)^2 \right] = \frac{1}{n}\sum_{t=2}^n \frac{1}{(t-1)^2} \sum_{s=1}^{t-1} \mathbb{E}\left[ \Delta_s^2 \right] \leq \frac{1}{n}\sum_{t=2}^n \frac{4C^{1/2}}{t-1} = O\left( \frac{\log n}{n} \right).
    \end{align*}
    Moving on to $T_3$, we note that for any $t$, by definition of $\theta$,
    \begin{align*}
        \theta - \frac{1}{N} \sum_{i \in O_{t-1}} y_i - \frac{1}{N} \sum_{i \notin O_{t-1}} \fhat^{(t-1)}(x_i) = \frac{1}{N} \sum_{i \notin O_{t-1}} \left( y_i - \fhat^{(t-1)}(x_i) \right)\\ = \mathbb{E}\left[ \frac{1}{N}\frac{y_{I_t} - \hat{f}^{(t-1)}(x_{I_t})}{q_t(I_t)} \mid \F_{t-1} \right].
    \end{align*}
    By Jensen's inequality, we deduce that for any $t$,
    \begin{align*}
        \mathbb{E}\left[ \left( \theta - \frac{1}{N} \sum_{i \in O_{t-1}} y_i - \frac{1}{N} \sum_{i \notin O_{t-1}} \fhat^{(t-1)}(x_i) \right)^2 \right] = \mathbb{E}\left[ \left( \mathbb{E}\left[ \frac{1}{N}\frac{y_{I_t} - \hat{f}^{(t-1)}(x_{I_t})}{q_t(I_t)} \mid \F_{t-1} \right] \right)^2 \right]\\
        \leq \mathbb{E}\left[ \mathbb{E}\left[ \left( \frac{1}{N}\frac{y_{I_t} - \hat{f}^{(t-1)}(x_{I_t})}{q_t(I_t)} \right)^2 \mid \F_{t-1} \right] \right] = \mathbb{E}\left[ \left( \frac{1}{N}\frac{y_{I_t} - \hat{f}^{(t-1)}(x_{I_t})}{q_t(I_t)} \right)^2 \right]\\ \leq \left( \mathbb{E}\left[ \left( \frac{1}{N}\frac{y_{I_t} - \hat{f}^{(t-1)}(x_{I_t})}{q_t(I_t)} \right)^4 \right] \right)^{1/2} \leq C^{1/2}.
    \end{align*}
    Thus, it also follows that, corresponding to $t=1$,
    \begin{align*}
        \mathbb{E}[T_3] = \mathbb{E}\left[ \frac{1}{n} \left( \theta - \frac{\sum_{i=1}^N \fhat^{(0)}(x_i)}{N} \right)^2 \right] \leq \frac{C^{1/2}}{n} = O\left( \frac{1}{n} \right).
    \end{align*}
    Now, working on $T_2$, by Cauchy-Schwarz, and using our above facts,
    \begin{align*}
        \mathbb{E}[T_2] = \frac{2}{n} \sum_{t=2}^n \mathbb{E}\left[ \left| \theta - \frac{\sum_{i \in O_{t-1}} y_i + \sum_{i \notin O_{t-1}} \fhat^{(t-1)}(x_i)}{N} \right| \left| \thetahat_{t-1} - \theta \right| \right]\\
        \leq \frac{2}{n} \sum_{t=2}^n \mathbb{E}\left[ \left( \theta - \frac{\sum_{i \in O_{t-1}} y_i + \sum_{i \notin O_{t-1}} \fhat^{(t-1)}(x_i)}{N} \right)^2 \right]^{1/2} \mathbb{E}\left[ \left( \thetahat_{t-1} - \theta \right)^2 \right]^{1/2}\\
        \leq \frac{2}{n}\sum_{t=2}^n \left( C^{1/2}\right)^{1/2}\left( \frac{4C^{1/2}}{t-1} \right)^{1/2} = \frac{4C^{1/2}}{n} \sum_{t=2}^n \frac{1}{\sqrt{t-1}} = O\left(  \frac{1}{\sqrt{n}} \right).
    \end{align*}
    Putting everything together, we deduce that
    \begin{align*}
        \mathbb{E}[|\Bhat_n - B_n|] \leq \mathbb{E}[T_1] + \mathbb{E}[T_2] + \mathbb{E}[T_3] = O\left( \frac{\log n}{n} \right) + O\left(  \frac{1}{\sqrt{n}} \right) + O\left( \frac{1}{n} \right) \to 0,
    \end{align*}
    which implies that $\Bhat_n \overset{L^1}{\to} B_n$, which implies $\Bhat_n - B_n = o_p(1)$. Thus, we conclude that $\sigmahat_n^2 \toP \sigma^2$, which by Slutsky's gives us our desired final result:
    \begin{align*}
        \sqrt{n}\left( \thetahat_n - \theta \right) / \sigmahat_{n} \toD \mathcal{N}\left( 0, 1\right).
    \end{align*}
\end{proof}

\subsection{Derivation and properties of the improved finite-sample variance estimator}
\label{appendix:bias-corrected-variance-estimator}

In this section, we motivate, derive, and provide commentary on the improved finite-sample variance estimator $\sigmahat^2_{n,\text{cor}}$ introduced in Section \ref{sec:bias-corrected-variance-estimator}.

\subsubsection{Inspirations from second-order Taylor expansions and the ABC and BCa bootstraps}

Before starting our derivation, recall from the proof of Theorem \ref{thm:pai-asymptotic-distribution} in Appendix \ref{proof:pai-thm} that we assumed the following variance stabilization to facilitate our martingale CLT:
\begin{align*}
    \frac{1}{n} \sum_{t=1}^n \Var(\Delta_t \mid \F_{t-1}) := A_n - B_n \toP \sigma^2,
\end{align*}
where $\Ahat_n$ and $\Bhat_n$ are the plug-in estimates for $A_n$ and $B_n$, respectively. From the aforementioned proof, we know that $\Ahat_n$ is unbiased for $A_n$ because $A_n - \Ahat_n$ is a martingale sequence. Thus, we elect to bias-correct by working on $\Bhat_n$, which is biased for $B_n$ because we plugged in $\thetahat_{t-1}$ in place of $\theta$. Specifically, we proceed to derive a correction term $\Chat_n$ via a second-order Taylor expansion $\Bhat_n - B_n$ around $\theta$, inspired philosophically by the ABC and BCa bootstrap methods (c.f.~\cite{efron1987better,diciccio1992more}). From the aforementioned proof, we know that
\begin{align*}
        \Bhat_n - B_n &= \frac{1}{nN^2} \sum_{t=2}^n \left( N\thetahat_{t-1} - \sum_{i \in O_{t-1}} y_i - \sum_{i \notin O_{t-1}} \hat{f}^{(t-1)}(x_i) \right)^2\\ 
        &- \frac{1}{n N^2} \sum_{t=2}^n \left( N\theta - \sum_{i \in O_{t-1}} y_i - \sum_{i \notin O_{t-1}} \fhat^{(t-1)}(x_i)  \right)^2\\
        &- \frac{1}{nN^2} \left( N\theta - \sum_{i=1}^N \fhat^{(0)}(x_i) \right)^2.
\end{align*}
Now, for $t=1,\ldots, n$, define $\psi_t := \frac{1}{N}\left( \sum_{i \in O_{t-1}} y_i + \sum_{i \notin O_{t-1}} \hat{f}^{(t-1)}(x_i) \right)$ and $g_t(\theta) = (\theta - \psi_t)^2$. Then, by substitution, we have the following, where $g_1(\theta)$ can be interpreted as a boundary term:
\begin{align*}
    \Bhat_n - B_n = \frac{1}{n} \sum_{t=2}^n \left( g_t(\thetahat_{t-1}) - g_t(\theta) \right) - \frac{1}{n} g_1(\theta).
\end{align*}
Taylor-expanding each $g_t$ around $\theta$, we have, after noting that $g'_t(\theta) = 2(\theta - \psi_t)$ and that $g''_t(\theta) = 2$, the following exact expression as the $g_t$ were originally quadratic:
\begin{align*}
    g_t(\thetahat_{t-1}) &= g_t(\theta) + g'_t(\theta)(\thetahat_{t-1} - \theta) + \frac{1}{2} g''_t(\theta)(\thetahat_{t-1} - \theta)^2\\
    &= g_t(\theta) + 2(\theta - \psi_t)(\thetahat_{t-1} - \theta) + (\thetahat_{t-1} - \theta)^2.
\end{align*}
Substituting, we have the following, with our boundary term having no $\thetahat_{t-1}$ analog:
\begin{align*}
    \Bhat_n - B_n = \underbrace{\frac{1}{n} \sum_{t=2}^n (\thetahat_{t-1} - \theta)^2}_{\text{second-order Taylor curvature term}} + \underbrace{\frac{1}{n} \sum_{t=2}^n 2(\theta - \psi_t)(\thetahat_{t-1} - \theta)}_{\text{first-order Taylor term}} - \underbrace{\frac{1}{n} g_1(\theta)}_{\text{boundary term}}.
\end{align*}
Because the purpose of a second-order Taylor correction is to correct the second-order curvature term, we thus explore a correction $\Chat_n$ that well-approximates $\frac{1}{n} \sum_{t=2}^n (\thetahat_{t-1} - \theta)^2$. In particular, we do not attempt to bias-correct the first-order Taylor term because it is (a) sign-indefinite; (b) generally not safely estimable without plugging in $\thetahat$ for $\theta$ (and risk introducing another layer of finite-sample bias); and (c) better treated as a component of the stochastic error.

\subsubsection{Constructing the correction}

Recall our notation in Theorem \ref{thm:pai-asymptotic-distribution} of
\begin{align*}
    \Delta_t := \phi_t - \theta, \quad \text{with} \quad \phi_t := \frac{1}{N}\left( \sum_{i \in O_{t-1}} y_i + \sum_{i \notin O_{t-1}} \hat{f}^{(t-1)}(x_{i}) \right) + \frac{1}{N}\frac{y_{I_t} - \hat{f}^{(t-1)}(x_{I_t})}{q_t(I_t)},
\end{align*}
recalling that $\thetahat_{t-1} := \frac{1}{t-1} \sum_{s=1}^{t-1} \phi_s$. Since $\thetahat_{t-1} = \theta + \frac{1}{t-1}\sum_{s=1}^{t-1}\Delta_s$ and the $\Delta_s$ are martingale differences, by orthogonality we have the following expectation of the target second-order Taylor curvature:
\begin{align}
\label{eq:ET1-decomp}
    \mathbb{E}[\frac{1}{n} \sum_{t=2}^n (\thetahat_{t-1} - \theta)^2] = \frac{1}{n}\sum_{t=2}^n \frac{1}{(t-1)^2}\sum_{s=1}^{t-1}\mathbb{E}[v_s] = \frac{1}{n}\sum_{s=1}^{n-1}\mathbb{E}[v_s]\,w_s, \quad \text{with} \quad w_s := \sum_{k=s}^{n-1}\frac{1}{k^2},
\end{align}
where the second equality exchanges the order of summation, and reusing computations from our proof of Theorem \ref{thm:pai-asymptotic-distribution} in Appendix \ref{proof:pai-thm}, we have per-step conditional variances $v_s$ as follows:
\begin{align*}
    v_s &:= \Var(\Delta_s \mid \F_{s-1})\\
    &= \frac{1}{N^2} \sum_{i \notin O_{s-1}} \frac{\left( y_i - \fhat^{(s-1)}(x_i) \right)^2}{q_s(i)} -  \frac{1}{N^2} \left( \sum_{i \notin O_{s-1}} \left( y_i - \fhat^{(s-1)}(x_i) \right) \right)^2\\
    &= \frac{1}{N^2} \sum_{i \notin O_{s-1}} \frac{\left( y_i - \fhat^{(s-1)}(x_i) \right)^2}{q_s(i)} - \left( \theta - \frac{1}{N} \sum_{i \in O_{s-1}} y_i -  \frac{1}{N} \sum_{i \notin O_{s-1}} \fhat^{(s-1)}(x_i) \right)^2\\
    &= \frac{1}{N^2} \sum_{i \notin O_{s-1}} \frac{\left( y_i - \fhat^{(s-1)}(x_i) \right)^2}{q_s(i)} - (\theta - \psi_{s})^2.\\
    &= \frac{1}{N^2} \sum_{i \notin O_{s-1}} \frac{\left( y_i - \fhat^{(s-1)}(x_i) \right)^2}{q_s(i)} - g_s(\theta).
\end{align*}
Ideally, our oracle correction term would just be $\Chat_n^* = \frac{1}{n} \sum_{s=1}^{n-1} v_s w_s$, but it is clear from above that $v_s$ is unknown as it involves $\theta$. To find a data-driven proxy for the above $v_s$, we note that letting
\begin{align*}
    \zeta_s = \frac{1}{N} \frac{y_{I_s} - \fhat^{(s-1)}(x_{I_s})}{q_s(I_s)},
\end{align*}
we have
\begin{align*}
    \mathbb{E}[\zeta_s^2 \mid \F_{s-1}] = \sum_{i \notin O_{s-1}} \left( \frac{1}{N} \frac{y_{i} - \fhat^{(s-1)}(x_i)}{q_s(i)} \right)^2 P(I_s = i \mid \F_{s-1})\\
    = \sum_{i \notin O_{s-1}} \left( \frac{1}{N} \frac{y_{i} - \fhat^{(s-1)}(x_i)}{q_s(i)} \right)^2 q_s(i) = \frac{1}{N^2} \sum_{i \notin O_{s-1}} \frac{\left( y_i - \fhat^{(s-1)}(x_i) \right)^2}{q_s(i)}.
\end{align*}
Thus, our derived additive correction is
\begin{align*}
    \Chat_n := \frac{1}{n}\sum_{s=1}^{n-1} \left( \frac{1}{N} \frac{y_{I_s} - \fhat^{(s-1)}(x_{I_s})}{q_s(I_s)} \right)^2 w_s, \quad \text{with} \quad w_s := \sum_{k=s}^{n-1}\frac{1}{k^2}.
\end{align*}
Our bias-corrected variance estimator is thus $\sigmahat^2_{n,\text{cor}} := \Ahat_n - \Bhat_n + \Chat_n$. As a follow-up, we do not directly target $g_s(\theta)$ for correction. Although $g_s(\theta)$ appears in the analytical second-order curvature expression, estimating it would require an additional plug-in approximation for $\theta$, which risks introducing a second layer of finite-sample bias. We therefore use $\Chat_n$ as a feasible positive surrogate: it targets the directly estimable component of the curvature correction and preserves consistency (see below).

\subsubsection{Preserving consistency and asymptotic normality}

To prove that $\sigmahat^2_{n,\text{cor}} \toP \sigma^2$, where $\sigma^2$ is the limiting average increment variance that governs the martingale CLT in Theorem \ref{thm:pai-asymptotic-distribution}, by Slutsky's it suffices to prove that $\Chat_n \toP 0$, as we already know that $\Ahat_n - \Bhat_n \toP \sigma^2$ from the corresponding proof in Appendix \ref{proof:pai-thm}. By Jensen's Inequality and Assumption \ref{a3:4th-moment-control}, we know that
\begin{align*}
    \mathbb{E}[\Chat_n] = \frac{1}{n}\sum_{s=1}^{n-1} \mathbb{E}\left[ \left( \frac{1}{N} \frac{y_{I_s} - \fhat^{(s-1)}(x_{I_s})}{q_s(I_s)} \right)^2 \right] w_s \leq \frac{C^{1/2}}{n}\sum_{s=1}^{n-1} w_s = \frac{C^{1/2}}{n}\sum_{s=1}^{n-1} \sum_{k=s}^{n-1}\frac{1}{k^2} \to 0.
\end{align*}
Because $\Chat_n \geq 0$, it follows that $\mathbb{E}[|\Chat_n|] \to 0$, which implies that $\Chat_n \overset{L_1}{\to} 0$, which implies that $\Chat_n = o_p(1)$. Immediately by Slutsky's Theorem, under the assumptions of Theorem \ref{thm:pai-asymptotic-distribution}, we have
\begin{align*}
    \sqrt{n}(\thetahat_n - \theta) / \sigmahat_{n,\text{cor}} \toD \mathcal{N}(0, 1).
\end{align*}

\subsubsection{An MSE reduction argument for $\sigmahat^2_{n,\text{cor}}$}

We proceed to characterize when $\sigmahat^2_{n,\text{cor}}$ will have a lower MSE than $\sigmahat^2_n$. Expanding and substituting above, we have the following decomposition after adding and subtracting $A_n - B_n$:
\begin{align*}
    \sigmahat^2_n - \sigma^2 = \Ahat_n - \Bhat_n - \sigma^2 = (\Ahat_n - A_n) - (\Bhat_n - B_n) + (A_n - B_n - \sigma^2)\\
    = (\Ahat_n - A_n) - \left( \frac{1}{n} \sum_{t=2}^n (\thetahat_{t-1} - \theta)^2 + \frac{1}{n} \sum_{t=2}^n 2(\theta - \psi_t)(\thetahat_{t-1} - \theta) - \frac{1}{n} g_1(\theta) \right) + (A_n - B_n - \sigma^2)\\
    = \underbrace{(\Ahat_n - A_n) + (A_n - B_n - \sigma^2) - \frac{1}{n} \sum_{t=2}^n 2(\theta - \psi_t)(\thetahat_{t-1} - \theta) + \frac{1}{n} (\theta - \psi_1)^2}_{\text{common error $E_n$}} - \underbrace{\frac{1}{n} \sum_{t=2}^n (\thetahat_{t-1} - \theta)^2}_{\text{targeted error $T_n$}}.
\end{align*}
Reusing the same scaffolding, for $\sigmahat^2_{n,\text{cor}} = \Ahat_n - \Bhat_n + \Chat_n$, we have
\begin{align*}
    \sigmahat^2_{n,\text{cor}} - \sigma^2 = E_n - T_n + \Chat_n.
\end{align*}
Squaring and simplifying, we have
\begin{align*}
    \text{MSE}(\sigmahat^2_{n,\text{cor}}) - \text{MSE}(\sigmahat^2_n) =  \mathbb{E}\left[ \left( E_n - T_n + \Chat_n \right)^2 \right] - \mathbb{E}\left[ \left( E_n - T_n \right)^2 \right]\\
    = \mathbb{E}\left[ \left( \Chat_n - T_n \right)^2 \right] + 2\mathbb{E}\left[ \Chat_n E_n \right] - \mathbb{E}\left[T_n^2\right].
\end{align*}
Finally, we have the following necessary and sufficient criteria for when $\text{MSE}(\sigmahat^2_{n,\text{cor}}) \leq \text{MSE}(\sigmahat^2_n)$:
\begin{align*}
    \text{MSE}(\sigmahat^2_{n,\text{cor}}) \leq \text{MSE}(\sigmahat^2_n) \iff \underbrace{\frac{\mathbb{E}\left[ \left( \Chat_n - T_n \right)^2 \right]}{\mathbb{E}\left[T_n^2\right]}}_{\alpha_n} + \underbrace{\frac{2\mathbb{E}\left[ \Chat_n E_n \right]}{\mathbb{E}\left[T_n^2\right]}}_{\beta_n} \leq 1.
\end{align*}
In context, $\Chat_n$ is our data-driven plug-in estimate for $T_n$. In words, the result above tells us that so long as (a) $\Chat_n$ predicts $T_n$ better than had we used $0$ (i.e., $\alpha_n$ small); and (b) our correction term $\Chat_n$ is not excessively\footnote{Relative to the magnitude of $\mathbb{E}[T_n^2]$} adversarially positively correlated with the common remaining error term $E_n$ (i.e., $\beta_n$ small), then we should expect $\sigmahat^2_{n,\text{cor}}$ to have a lower MSE for $\sigma^2$ than $\sigmahat^2_n$.

\subsection{Proof of Proposition \ref{prop:superpopulation-asymptotic-distribution}}
\label{appendix:proof-of-superpopulation-asymptotic-distribution}

Let $\thetahat_n$ be our PAI estimator for $\theta_N$ from Equation \ref{eqn:pai-estimator-for-mean}. We first observe that
\begin{align*}
    \sqrt{n}\left( \thetahat_n - \theta \right) = \sqrt{n}(\thetahat_n - \theta_N) + \frac{\sqrt{n}}{\sqrt{N}} \sqrt{N}(\theta_N - \theta).
\end{align*}
For any $s, t$, we have the following, noting that $\theta_N$ is a deterministic function of $y_1, \ldots, y_N$. Formally, the limits here are taken with respect to $(n, N) \uparrow \infty$ in Theorem \ref{thm:pai-asymptotic-distribution}'s triangular array setup, letting $\F_N = \sigma\left( \{ (x_i, y_i) \}_{i=1}^N \right)$:
\begin{align*}
    \lim P\left( \sqrt{n}(\thetahat_n - \theta_N) \leq s, \sqrt{N}(\theta_N - \theta) \leq t \right) = \lim \mathbb{E}\left[ \bdone\left( \sqrt{n}(\thetahat_n - \theta_N) \leq s, \sqrt{N}(\theta_N - \theta) \leq t \right) \right]\\
    = \lim \mathbb{E}\left[ \bdone\left( \sqrt{n}(\thetahat_n - \theta_N) \leq s\right)\bdone\left( \sqrt{N}(\theta_N - \theta) \leq t \right) \right]\\
    = \lim \mathbb{E}\left[ \mathbb{E}\left[ \bdone\left( \sqrt{n}(\thetahat_n - \theta_N) \leq s\right)\bdone\left( \sqrt{N}(\theta_N - \theta) \leq t \right) \mid \F_N \right] \right]\\
    = \lim \mathbb{E}\left[ \bdone\left( \sqrt{N}(\theta_N - \theta) \leq t \right) \mathbb{E}\left[ \bdone\left( \sqrt{n}(\thetahat_n - \theta_N) \leq s\right) \mid \F_N \right] \right]\\
    = \Phi\left( \frac{s}{\sigma} \right) \lim \mathbb{E}\left[ \bdone\left( \sqrt{N}(\theta_N - \theta) \leq t \right) \right] = \Phi\left( \frac{s}{\sigma} \right) \Phi\left( \frac{t}{\tau} \right),
\end{align*}
where $\sigma^2$ is the limit of the average conditional variances of the $\Delta_t$ (c.f.~Theorem \ref{thm:pai-asymptotic-distribution} and Assumption \ref{a1:variance_stabilization}). The last two equalities come from the asymptotic normality of $\thetahat_n$ proven in Theorem \ref{thm:pai-asymptotic-distribution}, the Dominated Convergence Theorem, and the standard Central Limit Theorem for~i.i.d. random variables. Because the joint CDF factors into two independent terms, we deduce that
\begin{align*}
    \left( \sqrt{n}(\thetahat_n - \theta_N), \sqrt{N}(\theta_N - \theta)\right) \toD \mathcal{N}\left( \bdzero, \diag\left( \sigma^2, \tau^2 \right) \right).
\end{align*}
By the Continuous Mapping Theorem, it follows that
\begin{align}
\label{eqn:asymptotic-distribution-superpop-appendix}
    \sqrt{n}\left( \thetahat_n - \theta \right) \toD \mathcal{N}\left( 0, \sigma^2 + \rho \tau^2 \right),
\end{align}
where $\rho := \lim n/N$ from Theorem \ref{thm:pai-asymptotic-distribution}.

\subsection{Proof of Theorem \ref{thm:pai-bai-var}}
\label{appendix:proof-of-anti-bai-thm}

From the proof of Theorem \ref{thm:pai-asymptotic-distribution} in Appendix \ref{proof:pai-thm}, we have the following after substituting in PAI's oracle sampling rule, with the expectation necessary because we have to average over all random sampling paths, and permissible because of orthogonality of martingale increments:
\begin{align*}
    \Var(\thetahat_{\text{PAI}}) = \frac{1}{n^2 N^2} \sum_{t=1}^n \mathbb{E}\left[ \left( \sum_{i \notin O_{t-1}} \frac{\left( y_i - f(x_i) \right)^2}{q_t(i)} \right) - \left( \sum_{i \notin O_{t-1}} \left( y_i - f(x_i) \right) \right)^2 \right]\\
    = \frac{1}{n^2 N^2} \sum_{t=1}^n \mathbb{E}\left[ \left( \sum_{i \in R_{t-1}} \frac{r_i^2}{q_t(i)} \right) - \left( \sum_{i \in R_{t-1}} e_i \right)^2 \right]\\
    = \frac{1}{n^2 N^2} \sum_{t=1}^n \mathbb{E}\left[ U(R_{t-1})^2 - S(R_{t-1})^2 \right] := \frac{1}{n^2 N^2} \sum_{t=1}^n \mathbb{E}\left[ U(R_{t-1})^2\left(1 - b_t^2 \right) \right]\\
    = \frac{1}{n^2 N^2} \sum_{t=1}^n \mathbb{E}\left[ U(R_{t-1})^2\right] - \mathbb{E}\left[ U(R_{t-1})^2 b_t^2\right],
\end{align*}
if we define $b_t$ as capturing the \textit{net signed direction} of the residual mass:
\begin{align*}
    b_t = \frac{S(R_{t-1})}{U(R_{t-1})} = \frac{\sum_{i \in R_{t-1}} y_i - f(x_i)}{\sum_{i \in R_{t-1}} |y_i - f(x_i)|}.
\end{align*}
To formalize, define $\bbar^2$ to be the \textit{weighted-average (squared) net signed direction} of the residual mass over the $n$ sampling rounds (weighted by $U(R_{t-1})^2$):
\begin{align}
\label{eqn:bbar2}
    \bbar^2 := \frac{\sum_{t=1}^n \mathbb{E}\left[ U(R_{t-1})^2 b_t^2\right]}{\sum_{t=1}^n \mathbb{E}\left[ U(R_{t-1})^2 \right]}.
\end{align}
Then, we have
\begin{align*}
    \Var(\thetahat_{\text{PAI}}) = \frac{1}{n^2 N^2} \sum_{t=1}^n \mathbb{E}\left[ U(R_{t-1})^2\right] - \bbar^2 \mathbb{E}\left[ U(R_{t-1})^2 \right] = \frac{1 - \bbar^2}{n^2 N^2} \sum_{t=1}^n \mathbb{E}\left[ U(R_{t-1})^2\right] := (1 - \bbar^2) A_{\text{PAI}}.
\end{align*}
Working on the expectation term, letting $I_t$ represent the index sampled at time $t$, letting $Q(R) = \sum_{i \in R} r_i^2$ (``quadratic"), and letting $C(R) = \sum_{i \in R} r_i^3$ (``cubic"), we have by LOTE the following:
\begin{align*}
    \mathbb{E}\left[ U(R_{t})^2 \mid R_{t-1} \right] = \mathbb{E}[\left( U(R_{t-1}) - r_{I_{t}} \right)^2 \mid R_{t-1}] = \mathbb{E}\left[ U(R_{t-1})^2 - 2 r_{I_t} U(R_{t-1}) + r_{I_t}^2 \mid R_{t-1} \right]\\
    = U(R_{t-1})^2 - 2 U(R_{t-1}) \sum_{i \in R_{t-1}} \frac{r_i^2}{U(R_{t-1})} + \sum_{i \in R_{t-1}} \frac{r_i^3}{U(R_{t-1})} = U(R_{t-1})^2 - 2Q(R_{t-1}) + \frac{C(R_{t-1})}{U(R_{t-1})}.
\end{align*}
By Tower, we have the following recursion:
\begin{align*}
    \mathbb{E}\left[ U(R_{t})^2 \right] = \mathbb{E}[U(R_{t-1})^2] - \mathbb{E}\left[2 Q(R_{t-1}) - \frac{C(R_{t-1})}{U(R_{t-1})}\right]
\end{align*}
Applying the recursion $t$ times, and letting $R_0 = \mathcal{I}$, we have
\begin{align*}
    \mathbb{E}\left[ U(R_{t})^2 \right] = U(R_0)^2 - \sum_{s=0}^{t-1} \mathbb{E}\left[2 Q(R_{s}) - \frac{C(R_{s})}{U(R_{s})}\right],
\end{align*}
which tells us the following, after switching the second order of summation and reindexing:
\begin{align*}
    A_{\text{PAI}} = \frac{1}{n^2N^2} \sum_{t=1}^n \mathbb{E}[U(R_{t-1})^2] = \frac{1}{n^2N^2} \sum_{t=0}^{n-1} \mathbb{E}[U(R_{t})^2]\\
    = \frac{1}{n^2 N^2}\sum_{t=0}^{n-1} \left( U(R_0)^2 - \sum_{s=0}^{t-1} \mathbb{E}\left[2 Q(R_{s}) - \frac{C(R_{s})}{U(R_{s})}\right] \right)\\
    = \frac{1}{n N^2} U(R_0)^2 - \frac{1}{n^2 N^2} \sum_{s=0}^{n-2} (n - s - 1) \mathbb{E}\left[2 Q(R_{s}) - \frac{C(R_{s})}{U(R_{s})}\right].
\end{align*}
Now, because $Q(R_s) \leq Q(R_0)$ and because $C(R_s) / U(R_s) \geq 0$ for all $s$, we know by property of arithmetic sums that
\begin{align*}
    A_{\text{PAI}} \geq \frac{1}{n N^2} U(R_0)^2 - \frac{1}{n^2 N^2} \sum_{s=0}^{n-2} (n - s - 1) 2 Q(R_0)\\
    = \frac{1}{n N^2} U(R_0)^2 - \frac{ Q(R_0)(n-1)}{nN^2} = \frac{1}{N^2} \left( \frac{U(R_0)^2}{n} - \frac{n-1}{n} Q(R_0) \right).
\end{align*}

Computing BAI/SAI's variance (c.f. Equation \ref{eqn:bai-sai}), we have
\begin{align*}
    \Var(\thetahat_{\text{active}}) = \frac{1}{N^2} \sum_{t=1}^N \frac{(y_t - f(x_t))^2 (1- \pi_t(x_t))}{\pi_t(x_t)} = \frac{1}{N^2} \sum_{t=1}^N r_t^2 \left( \frac{1}{\pi_t(x_t)} - 1 \right).
\end{align*}

After substituting in SAI/BAI's oracle policy, we have:
\begin{align*}
    \Var(\thetahat_{\text{active}}) = \frac{1}{N^2} \sum_{t=1}^N r_t^2 \left( \frac{1}{\pi_t(x_t)} - 1 \right) = \frac{1}{N^2} \left( \frac{U(R_0)^2}{n} - Q(R_0) \right).
\end{align*}
This is very convenient, because it tells us that
\begin{align*}
    A_{\text{PAI}} \geq \Var(\thetahat_{\text{active}}) + \frac{Q(R_0)}{nN^2}.
\end{align*}
Working on an upper bound, recall that
\begin{align*}
    A_{\text{PAI}} = \frac{1}{n N^2} U(R_0)^2 - \frac{1}{n^2 N^2} \sum_{s=0}^{n-2} (n - s - 1) \mathbb{E}\left[2 Q(R_{s}) - \frac{C(R_{s})}{U(R_{s})}\right].
\end{align*}
Because only one unit is removed at each timestep, letting $r_{\text{max}} = \max_{i \in R_0} r_i$, we have
\begin{align*}
    Q(R_{s}) \geq Q(R_0) - sr^2_{\text{max}}.
\end{align*}
Moreover,
\begin{align*}
    \frac{C(R_{s})}{U(R_{s})} = \frac{\sum_{i \in R_s} r_i^3}{\sum_{i \in R_s} r_i} \leq \frac{r^2_{\text{max}} \sum_{i \in R_s} r_i}{\sum_{i \in R_s} r_i} = r^2_{\text{max}}.
\end{align*}
Thus, we have
\begin{align*}
    A_{\text{PAI}} \leq \frac{1}{n N^2} U(R_0)^2 - \frac{1}{n^2 N^2} \sum_{s=0}^{n-2} (n - s - 1) \left( 2Q(R_0) - (2s + 1)r^2_{\text{max}}\right)\\
    = \frac{1}{N^2} \left( \frac{U(R_0)^2}{n} - \frac{n-1}{n} Q(R_0) \right) + \frac{r^2_{\text{max}}}{n^2 N^2} \sum_{s=0}^{n-2} (n-s-1)(2s+1).
\end{align*}
Working on the last summation, we have
\begin{align*}
    \sum_{s=0}^{n-2} (n-s-1)(2s+1) = \sum_{s=0}^{n-2} (n-1)(2s+1) - s(2s+1)\\
    = (n-1)\sum_{s=0}^{n-2} (2s+1) - 2 \sum_{s=0}^{n-2} s^2 - \sum_{s=0}^{n-2} s = (n-1)^3 - \left( 2 \sum_{s=0}^{n-2} s^2 + \sum_{s=0}^{n-2} s\right) = \frac{n(n-1)(2n-1)}{6}.
\end{align*}
Thus, we conclude that
\begin{align*}
    A_{\text{PAI}} \leq \frac{1}{N^2} \left( \frac{U(R_0)^2}{n} - \frac{n-1}{n} Q(R_0) \right) + \frac{(n-1)(2n-1) r^2_{\text{max}}}{6 nN^2}\\
    = \Var(\thetahat_{\text{active}}) + \frac{Q(R_0)}{nN^2} + \frac{(n-1)(2n-1) r^2_{\text{max}}}{6 nN^2}.
\end{align*}
Putting everything together, we deduce that
\begin{align*}
    \Var(\thetahat_{\text{active}}) + \frac{Q(R_0)}{nN^2} \leq A_{\text{PAI}} \leq \Var(\thetahat_{\text{active}}) + \frac{Q(R_0)}{nN^2} + \frac{(n-1)(2n-1) r^2_{\text{max}}}{6 nN^2},
\end{align*}
which implies that
\begin{align*}
    (1 - \bbar^2) \left( \Var(\thetahat_{\text{active}}) + \frac{Q(R_0)}{nN^2} \right) \leq \Var(\thetahat_{\text{PAI}})\\ \leq (1 - \bbar^2) \left( \Var(\thetahat_{\text{active}}) + \frac{Q(R_0)}{nN^2} + \frac{(n-1)(2n-1) r^2_{\text{max}}}{6 nN^2} \right),
\end{align*}
or, more suggestively, recalling that (i) $Q(R_0) = \sum_{i=1}^N (y_i - f(x_i))^2$; (ii) $r_{\text{max}} = \max_{i \in [1:N]} |y_i - f(x_i)|$; and (iii) $\bbar^2$ is the weighted-average (squared) net signed direction of the residual mass over the $n$ sampling rounds (c.f. Equation \ref{eqn:bbar2}):
\begin{align*}
    \frac{1}{1 - \bbar^2} \cdot \frac{1}{1 + \frac{Q(R_0)}{nN^2 \Var(\thetahat_{\text{active}})} + \frac{(n-1)(2n-1) r^2_{\text{max}}}{6 nN^2 \Var(\thetahat_{\text{active}})}} \leq \frac{\Var(\thetahat_{\text{active}})}{\Var(\thetahat_{\text{PAI}})} \leq \frac{1}{1 - \bbar^2} \frac{1}{1 + \frac{Q(R_0)}{nN^2 \Var(\thetahat_{\text{active}})}},
\end{align*}
or, with one final substitution of $\Var(\thetahat_{\text{active}}) = \frac{1}{N^2} \left( \frac{U(R_0)^2}{n} - Q(R_0) \right)$,
\begin{align*}
    \frac{1}{1 - \bbar^2} \cdot \frac{1}{1 + \frac{Q(R_0)}{U(R_0)^2 - n Q(R_0)} + \frac{(n-1)(2n-1) r^2_{\text{max}}}{6 \left( U(R_0)^2 - n Q(R_0) \right)}} \leq \frac{\Var(\thetahat_{\text{active}})}{\Var(\thetahat_{\text{PAI}})} \leq \frac{1}{1 - \bbar^2} \frac{1}{1 + \frac{Q(R_0)}{U(R_0)^2 - n Q(R_0)}}.
\end{align*}
Formally, when $n r_{\text{max}} / U(R_0) \to 0$,\footnote{This is just one possible sufficient condition.} both remainder terms $\frac{Q(R_0)}{U(R_0)^2 - n Q(R_0)}$ and $\frac{(n-1)(2n-1) r^2_{\text{max}}}{6 \left( U(R_0)^2 - n Q(R_0) \right)}$ vanish asymptotically\footnote{This follows from algebra as $Q(R_0) \leq r_{\text{max}} U(R_0)$.}, yielding the following asymptotic relative-efficiency relationship:
\begin{align*}
    \frac{\Var(\thetahat_{\text{active}})}{\Var(\thetahat_{\text{PAI}})} = \frac{1 + o(1)}{1 - \bbar^2}.
\end{align*}
The sufficient condition above admits an intuitive decomposition:
\begin{align*}
    \frac{n r_{\text{max}}}{U(R_0)} = \frac{n}{N} \cdot \frac{r_{\text{max}}}{\bar{r}}, \quad \text{where} \quad \bar{r} := \frac{1}{N} U(R_0) = \frac{1}{N} \sum_{i=1}^N \left| y_i - f(x_i) \right|.
\end{align*}
Since $r_{\text{max}} \geq \bar{r}$, this condition necessarily requires $n/N \to 0$, which is permitted by Theorem \ref{thm:pai-asymptotic-distribution}. More precisely, the labeling fraction must vanish sufficiently quickly relative to $r_{\text{max}} / \bar{r}$. In particular, $n/N \to 0$ suffices if $r_{\text{max}} / \bar{r} = O(1)$.

\subsection{Diagonal stochastic Newton updates for online logistic regression}
\label{appendix:olr-updates}

The following is adapted from Equations 3.1-3.3 of \cite{bercu2020efficient}, who maintain a dense $D \times D$ curvature matrix, whereas we use a diagonal approximation.

Let $\beta_t \in \R^D$ denote our logistic regression coefficients at round $t$, and let $h_t \in \R^D$ denote our accumulated diagonal curvature vector. For all datasets except for AlphaFold (which provides fixed probabilistic ML predictions without model access) and the two LLM evaluation suites MMLU-Pro and BBH+GPQA+IFEval+MATH+MusR, we initialize $\beta_0 = \beta_{\text{init}}$ by fitting on a training split of the data and initialize $h_0 = \bdone_D$. For the two LLM evaluation suites, we use the historical-data regimen described in \cite{wu2026efficient}: fitted question factors are held fixed, while $\beta_0$ and $h_0$ are initialized to the empirical mean and inverse marginal variances, respectively, of the historical LLMs' fitted model factors. These initializations are identical across all new LLMs.

After observing $(x_{I_t}, y_{I_t})$, we compute
\begin{align}
    h_t = h_{t-1} + \sigma\left( \beta_{t-1}^\top x_{I_t} \right) \left( 1 - \sigma\left( \beta_{t-1}^\top x_{I_t}\right) \right) x_{I_t}^{\odot 2},\\ \beta_t = \beta_{t-1} + (y_{I_t} - \sigma\left( \beta_{t-1}^\top x_{I_t} \right)) x_{I_t} \oslash h_t,
    \label{eqn:olr-updates}
\end{align}
where $\odot 2$ represents elementwise squaring, and $\oslash$ represents elementwise division.

\newpage

\section{Additional results}

\begin{figure}[!htbp]
    \centering
    \includegraphics[width=1.0\linewidth]{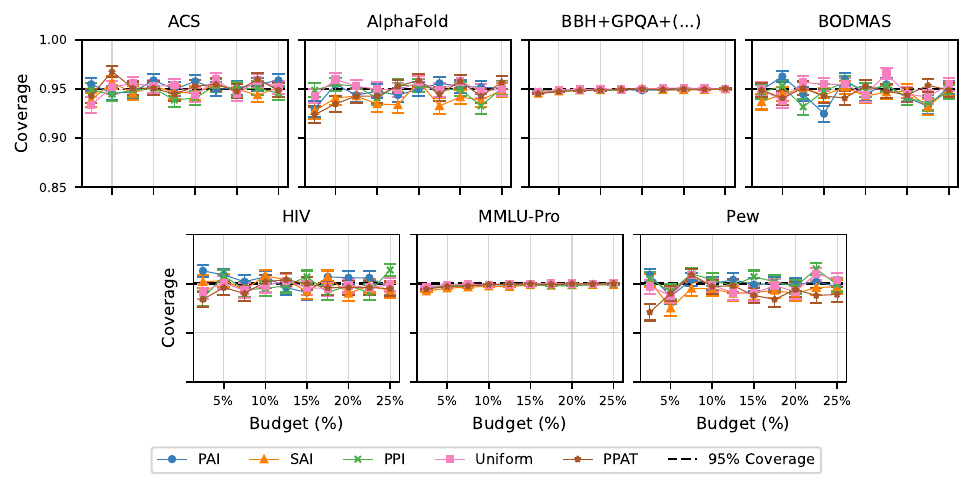}
    \caption{\textbf{Comparing empirical coverages of Pro-Active Inference (PAI), Sequential Active Inference (SAI), Prediction-Powered Inference (PPI), Uniform Sampling (Uniform), and Prediction-Powered Active Testing (PPAT) across seven datasets.} Each subplot shows the empirical coverage of nominal $95\%$ normal-approximation confidence intervals as a function of labeling budget. For all datasets except the two LLM evaluation suites (MMLU-Pro and BBH+GPQA+IFEval+MATH+MusR), empirical coverages are computed over $1000$ random seeds. For the two LLM evaluation suites, empirical coverages are averaged over $2.2$K test models per evaluation suite and $100$ seeds for computational efficiency. Standard errors are shown, but are negligible at the plotted scale.}
    \label{fig:all-coverages}
\end{figure}

\begin{figure}[!htbp]
    \centering
    \includegraphics[width=1.0\linewidth]{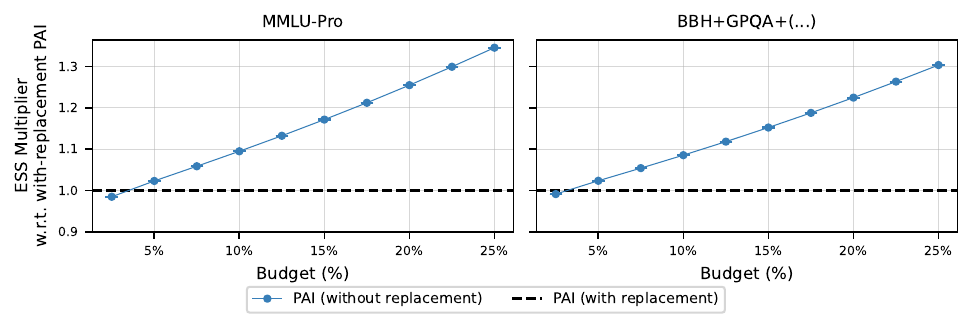}
    \caption{\textbf{ESS multipliers of without-replacement PAI over with-replacement PAI across the two LLM evaluation suites.} Each subplot shows the ESS multiplier of without-replacement PAI with respect to with-replacement PAI as a function of labeling budget. For example, on MMLU-Pro at a $22.5\%$ budget, an ESS multiplier of $1.3$ means that without-replacement PAI matches the confidence interval width of with-replacement PAI at a $29.25\%$ budget. Standard errors are shown, but are negligible at the plotted scale. \textbf{Higher ESS multipliers indicate better sampling efficiency, with values above $1$ indicating that without-replacement PAI outperforms with-replacement PAI.}}
    \label{fig:pai-with-vs-without-replacement}
\end{figure}

\begin{sidewaysfigure}
    \centering
    \includegraphics[width=1.0\linewidth]{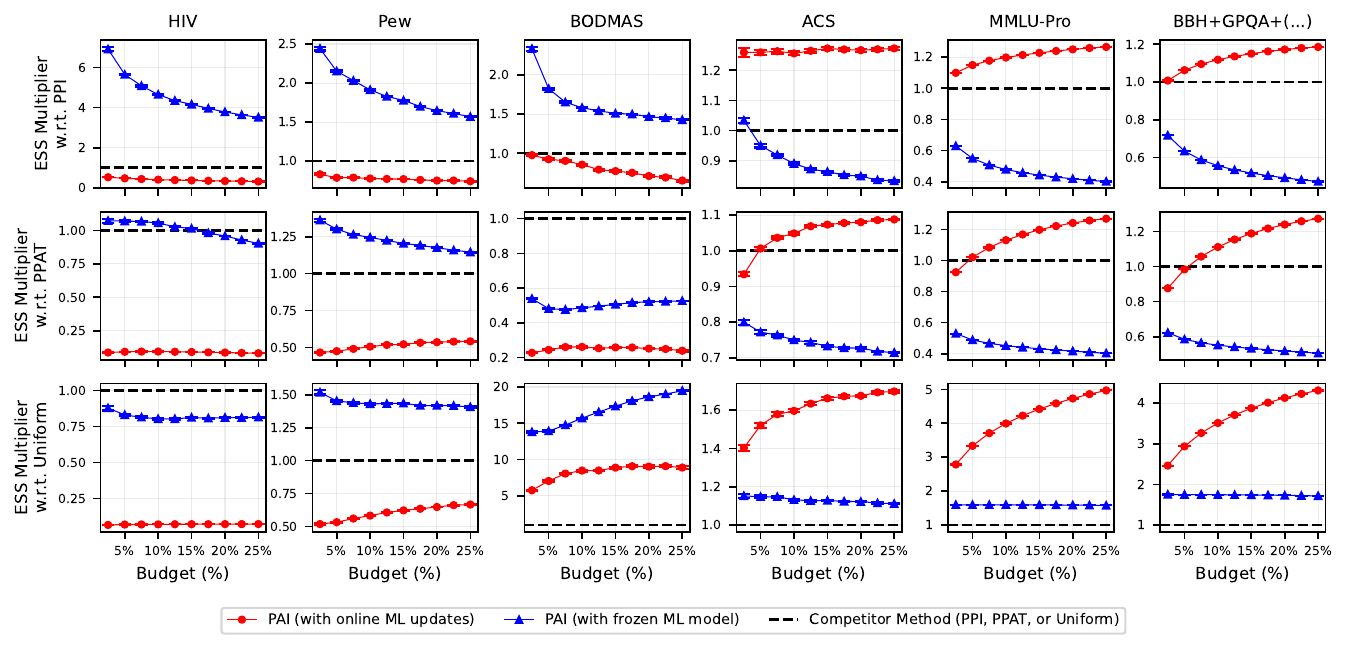}
    \caption{\textbf{ESS multipliers of PAI and its frozen-ML ablation over PPI, PPAT, and uniform sampling across six datasets.} Each row corresponds to one competitor (PPI, PPAT, or Uniform) and each column corresponds to one dataset. Each subplot shows the ESS multipliers of standard PAI (blue triangles) and frozen-ML PAI (red circles) with respect to the corresponding competitor method as a function of labeling budget. In the frozen-ML ablation, the initial ML model is held fixed throughout sampling, affecting both imputation and acquisition. The competing methods remain unchanged from their specifications in Section \ref{sec:experimental-setup}. \textbf{Higher ESS multipliers indicate better sampling efficiency, with values above $1$ indicating that the corresponding PAI variant outperforms the competitor.}}
    \label{fig:frozen-pai-ablation-all-datasets}
\end{sidewaysfigure}

\end{document}